\documentclass{birkjour_t2}
 \newtheorem{thm}{Theorem}[section]
 
 \newtheorem{lem}[thm]{Lemma}
 \newtheorem{prop}[thm]{Proposition}
 \theoremstyle{definition}
 \newtheorem{defn}[thm]{Definition}
 \theoremstyle{remark}

 \numberwithin{equation}{section}

\begin{document}

%
%
%
%
%
%
%
%
%

\title[Deltoid tangents and random tilings]{Deltoid tangents with evenly distributed orientations and random tilings}

\author[J.G.Escudero]{Juan Garc\'{\i}a Escudero}
\address{%
Madrid, Spain
}
\email{jjgemplubg@gmail.com}



\date{October, 15 2019}
\dedicatory{}

\begin{abstract}
We study the construction of substitution tilings of the plane based on certain simplicial configurations of tangents of the deltoid with evenly distributed orientations.  The random tiling ensembles are obtained as a result of tile rearrangements in the substitution rules associated to edge flips. 
\end{abstract}

\maketitle
\section{Introduction}
\bigskip\par
Tiling theory developed from several directions in the last half century. Interested in problems of logic, Wang \cite{wan61} and Berger \cite{ber66} produced  in the 1960s non periodic tilings of the plane with a large number of basic geometric shapes. Later Penrose generated a non periodic planar tiling by using a set of two basic tiles \cite{pen74}. The discovery of quasicrystals by Shechtman in the early 1980s  \cite{she84} stimulated the study of tilings with non crystallographic symmetries, because certain tilings underlie the atomic structure of quasicrystals. Two main approaches to modelling the structure of quasicrystals are  the construction of cut and project sets, also called model sets, and the generation of substitution tilings. Much progress has been done in the field of deterministic substitution tilings during the last decades, but not in the construction of random tilings.
 We will consider the problem of the generation of tilings by means of both deterministic and random substitutions. 
\par
In 1996 Nischke and Danzer provided a construction of substitution tilings based on the system of certain $d$ tangents of the deltoid, when $d$ is odd, higher than 5 and non divisible by three \cite{nis96}. In recent papers we have shown that the sets of deltoid tangents with $d$ evenly distributed orientations are images under a certain map of periodic billiard trajectories inside the fundamental region of the affine Weyl group of the root system $A_{2}$ \cite{esc17}. In  some cases the arrangements of deltoid tangents are simplicial, namely, all their bounded cells are triangles. Several sets of simplicial arrangements corresponding to $d$ odd and non divisible by 3 were used in \cite{nis96} to produce tilings, although the same sets were studied before with a different purpose \cite{fur84}. In this paper we study a construction of tilings by using deltoid tangents having any $d>4$ evenly distributed orientations and we use them to investigate the generation of random tilings for $d$ even.

\par
The paper is organised as follows. In Section 2 we study several properties of the simplicial arrangements of deltoid tangents with $d$ evenly distributed orientations. They contain the prototiles, or elementary triangles, and scaled copies of them necessary for the derivation of the inflation or substitution rules. The existence of deterministic substitution rules producing planar tilings is analysed in Section 3. The relationship between algebraic numbers and mathematical quasicrystals was established first by Meyer in \cite{mey70, mey72} and later in \cite{bom87}. Pisot-Vijayaraghavan (PV) numbers are real algebraic integers with absolute value larger than one whose conjugates lie within the unit circle in the complex plane. Most of the early studied tilings, including those found by Penrose and Ammann \cite{gru87, pen74}, were based upon inflation rules and have PV numbers as inflation factors.  A significant feature of some of the tilings obtained with the construction given in Section 3 is that many of the inflation factors turn out to be PV numbers (\cite{esc08}, Table 1), which is of interest form the point of view of the diffraction patterns produced when modelling physical quasicrystals with geometric tilings.  Examples of tilings associated with the simplicial arrangements corresponding to $d=5, 6, 8, 9,10,12$ have been analysed in previous works (see \cite{esc17} and references within) where we have shown that they contain vertex configurations with local rotational symmetry, a property not appearing in the tilings obtained with the method employed in \cite{nis96}, as remarked in \cite{gah15}. Although they are more complicated, to illustrate the general results we include in this work some examples corresponding to $d=14$. Section 4 is devoted to the study of the existence of random tilings related with the introduction of tile rearrangements in a given tiling or in the substitution rules. They are associated to edge flips in some quadrilaterals formed by the union of two prototiles. It was treated for several special cases in \cite{esc08, esc11, esc17} and we show that these results can be generalised.
 
 \section{Simplicial arrangements of deltoid tangents with evenly distributed orientations}
\bigskip\par
We consider the lines $L_{d,\kappa,\nu}(x,y)=0$, $d=5,6...,\kappa \in \{-2, 0, 2\}, \nu \in I^ {(\kappa)}$, having parametric equations $x= x(t), y= y(t)$, where

 \begin{equation}    
x(t)+\sqrt{-1} y(t):=e^{-2\pi \frac{3\nu-\kappa}{3d}\sqrt{-1} }+t e^{ \pi \frac{3\nu-\kappa}{3d}\sqrt{-1} }, t \in {\Bbb{R}},
 \end{equation}

\noindent
with  $I^ {(\kappa)}=\{0,1,2..., d-1\}$, for $\kappa \in \{-2, 0\}$ and $I^ {(2)}=\{0, -1, -2..., -d+1\}$. The construction is based on the configurations of $d$ lines 
   \begin{equation}   
   \mathcal{A}^{(\kappa)}_{d}:=\{L_{d,\kappa,\nu}(x,y)=0 \}_{\nu \in I^ {(\kappa)}}
    \end{equation} 
   \par\noindent 
which are defined $\forall d$ if $\kappa=0$ and for $d=3q, q=2, 3,...$ when $\kappa=\mp 2$. The configuration corresponding to $d=14$ is represented in Fig.1, together with a superimposed heptagon which will be necessary in Section 4 for the construction of random tilings.
     \par\noindent  
The lines are tangents of the deltoid $\mathcal{D}$, which has parametric equation
 \begin{equation}    
z(\varphi)=2 e^{\varphi \sqrt{-1} }+ e^{ -2 \varphi \sqrt{-1} },0 \leq \varphi \leq 2\pi
 \end{equation}
 
 If $\varphi=\frac{3\nu-\kappa}{3d}\pi $ then the tangency point of $L_{\varphi}(x,y):=y+({\rm cos}2\varphi-x){\rm tan}\varphi+{\rm sin}2\varphi=0$ to $\mathcal{D}$ is $z(-2\varphi)$  and it intersects $\mathcal{D}$ also in $z(\varphi)$ and $z(\varphi+\pi)$. We define the segments
 \begin{equation}    
G(\varphi):=\overline{z(\varphi);z(\varphi+\pi)}
 \end{equation}  
     \par\noindent  
and we write $z_{\nu}, G^{(\kappa)}_{\nu,d}, s_{\nu}$ instead of $z( \frac{(3\nu-\kappa)\pi}{3d}), G(\frac{(3\nu-\kappa)\pi}{3d}), {\rm sin}(\frac{\nu\pi}{d})$ respectively. The length of the segment $G^{(\kappa)}_{\nu,d}$ (the index $\nu$ is taken mod $d$) is $|G^{(\kappa)}_{\nu,d}|=4$ and the tangency point of $G^{(\kappa)}_{\nu,d}$ at the deltoid $\mathcal{D}$ is $z_{2(d-\nu)+\kappa}$. 
  \par
We use the following well-known results about congruences:
 
 \begin{prop}
If $\gcd(a,d)=1$, then the linear congruence $ax \equiv b$ (mod $d$) has exactly one solution modulo $d$.
\end{prop}
By using Euclid{'}s algorithm it can be shown that if $\gcd(a,d)=m$ then there are two integers $s$ and $t$ such that $m=as +dt$ (B\'ezout{'}s identity). Therefore when $m=1$, the solution of $ax \equiv b$ (mod $d$) is $x=bs$.

 \begin{prop}
If $\gcd(a,d)=m$, then the congruence $ax \equiv b$ (mod $d$) has a solution iff  $m|b$. In that case there are exactly $m$ solutions modulo $d$ which can be written as $x_{1}, x_{1}+d_{1},..., x_{1}+(m-1)d_{1}$, where $d=md_{1}$ and $x_{1}$ is the solution of the congruence $a_{1}x \equiv b_{1}$ (mod $d_{1}$), $a=ma_{1}$, $b=mb_{1}$. 
\end{prop} 
In what follows let $d>4$. The case $d=5$ has been studied in \cite{esc11}, where we have obtained the minimal first cohomology for a quasicrystal tiling space (see \cite{sad08} for an introduction to the study of the cohomology of tiling spaces). Other cases have been treated in \cite{esc08} and references therein.
 \par
The tilings we want to study have triangles as basic building blocks. Their edges lie on segments $G(\varphi), G( \psi)$ and $G(\chi)$.
 \par
\begin{lem}
 \par
 If $ 0\leq \varphi < \chi  < \psi < \pi$ and $\omega:=\varphi + \chi + \psi $ then 
  \begin{equation}
|G(\varphi)\cap G( \psi)-G(\varphi)\cap G(\chi)|=4 |{\rm sin} \omega|\cdot {\rm sin} (\psi-\chi)
\end{equation} 
\end{lem}
\begin{proof}
The point  $G(\varphi)\cap G(\psi)$ is
  \begin{equation}
p(\varphi, \psi):=(3 - 4({\rm sin}^2\varphi - {\rm sin}^2\varphi {\rm sin}^2\psi+ {\rm sin}^2\psi +
      {\rm cos}\varphi {\rm cos}\psi {\rm sin}\varphi {\rm sin}\psi), -4{\rm sin}\varphi
 {\rm sin}\psi{\rm sin}(\varphi+\psi))
 \end{equation}
 and  
    \begin{equation}
|p(\varphi, \psi)-p(\varphi, \chi)|=4|{\rm sin}(\varphi+\chi+\psi)|{\rm sin}(\psi-\chi)
 \end{equation}
  \par
A consequence of Eq. (2.7) is that if $\varphi+\chi+\psi \equiv 0$ (mod $\pi$) then $G(\varphi)\cap G(\chi)\cap G(\psi)\neq \emptyset$.
\end{proof} 
 \par

\par
\begin{defn}
A triangle $\Delta_{d}^{(\kappa)}( \lambda , \mu, \nu)$ formed by $G^{(\kappa)}_{\lambda,d}, G^{(\kappa)}_{\mu,d},  G^{(\kappa)}_{\nu,d}$ is said to be elementary if $ \lambda +\mu + \nu  \equiv \kappa \pm1$ (mod $d$), $\kappa \in \{-2, 0, 2\}$.
\end{defn}
 \par
\begin{lem}
 \par
 If $ 0\leq \lambda < \mu < \nu < d$, $ \lambda , \mu, \nu \in {\Bbb{N}}\cup \{0\}$ and $\sigma_{d} :=  \lambda +\mu + \nu $ then $G^{(\kappa)}_{\lambda,d}, G^{(\kappa)}_{\mu,d},  G^{(\kappa)}_{\nu,d}$ are either concurrent or form a triangle $\Delta_{d}^{(\kappa)}( \lambda , \mu, \nu)$ congruent to $\frac{s_{p}}{s_{1}}\cdot t_{d}$, where $s_{p}=|s_{\sigma_{d} -\kappa}|>0$ and $t_{d}$ is an elementary triangle.
\end{lem}

\begin{proof}
  \par
If $\lambda +\mu + \nu \not\equiv \kappa$ (mod $d$) then $G^{(\kappa)}_{\lambda,d}, G^{(\kappa)}_{\mu,d},  G^{(\kappa)}_{\nu,d}$ are not concurrent and form a triangle $\Delta_{d}^{(\kappa)}(\lambda , \mu, \nu)$ with angles $(\nu-\mu)\frac{\pi}{d}, (\mu-\lambda)\frac{\pi}{d}, (\lambda-\nu+d)\frac{\pi}{d}$. By Eq.(2.5) we have $|G^{(\kappa)}_{\lambda,d}\cap G^{(\kappa)}_{\nu,d}-G^{(\kappa)}_{\lambda,d}\cap G^{(\kappa)}_{\mu,d}|=4 |s_{\sigma_{d} -\kappa}| s_{\nu-\mu}$ and the lengths of the sides lying on $G^{(\kappa)}_{\lambda,d}, G^{(\kappa)}_{\mu,d}$ and $G^{(\kappa)}_{\nu,d}$ are $ 4 |s_{\sigma_{d} -\kappa}| s_{\nu-\mu}, 4 |s_{\sigma_{d} -\kappa}| s_{\lambda-\nu+d}$ and $4 |s_{\sigma_{d} -\kappa}| s_{\mu-\lambda}$ respectively.
\par
A segment with length $4 s_{1}s_{\nu}$ is denoted by  $S_{\nu}$. If $\Delta_{d}^{(\kappa)}( \lambda , \mu, \nu)$ is elementary, then its sides are $S_{\nu-\mu}, S_{\lambda-\nu+d}, S_{\mu-\lambda}$. If $p:=\sigma_{d} -\kappa \not\equiv \pm1$ (mod $d$) then the lengths of the sides of $\Delta_{d}^{(\kappa)}( \lambda , \mu, \nu)$  are those of an elementary triangle $t_{d}$ multiplied by $\iota_{d,p}:= \frac{s_{p}}{s_{1}}, 1 < p \leq q$.
\end{proof} 

         \begin{table}
\begin{tabular}{lllllll}
\hline \\
$d $&$q$&$\kappa$&$|\mu |$&$v_{2}(G^{(\kappa)}_{\mu, d})$ &$v_{3}(G^{(\kappa)}_{\mu, d})$ \\ 
\hline
$2q+1$&$\ne 3l+1$&$0$&$0$&$0$&$\frac{d-1}{2}$ \\
\hline
$2q+1$&$\ne 3l+1$&$0$&$\ne 0$&$2$&$\frac{d-3}{2}$ \\
\hline
$2q+1$&$3l+1$&$0$&$\in\{0, 2l+1, 4l+2\}$&$0$&$\frac{d-1}{2}$ \\
\hline
$2q+1$&$3l+1$&$0$&$\notin \{0, 2l+1, 4l+2\}$&$2$&$\frac{d-3}{2}$ \\
\hline
$2q$&$\ne 3l$&$0$&$\in I_{odd}\cup \{0\}$&$1$&$\frac{d-2}{2}$ \\
\hline
$2q$&$\ne 3l$&$0$&$\in I_{even}$&$3$&$\frac{d-4}{2}$ \\
\hline
$2q$&$3l$&$0$&$\in I_{odd}\cup \{0, 2l, 4l \}$&$1$&$\frac{d-2}{2}$ \\
\hline
$2q$&$3l$&$0$&$\notin I_{odd}\cup \{0, 2l, 4l \}$&$3$&$\frac{d-4}{2}$ \\
\hline
$3q$&$2l$&$\mp2$&$\in I_{odd}$&$1$&$\frac{d-2}{2}$ \\
\hline
$3q$&$2l$&$\mp2$&$\in I_{even}\cup \{0\}$&$3$&$\frac{d-4}{2}$ \\
\hline
$3q$&$2l+1$&$\mp2$&$\forall$&$2$&$\frac{d-3}{2}$ \\

\hline \\
\end{tabular}
\caption{Vertex multiplicities in $G^{(\kappa)}_{\mu, d}$, with  $l=1,2,3,....$}
\end{table} 

 We define the sets of indexes $I_{odd}:=\{1,3,5,... 2\lfloor{\frac{d}{2}}\rfloor-1\}, I_{even}:=\{2,4,6,... 2\lceil{\frac{d}{2}}\rceil-2\}$. The number of vertices  with multiplicity (number of segments incident to a vertex) $j$ belonging to the segment $G$ is denoted by $v_{j}(G)$.

\begin{lem}
 \par
Each set $\{G^{(\kappa)}_{\mu, d}\}_{\mu \in I^ {(\kappa)} }$ makes up a triangular pattern $\mathcal{G}_{\Delta,d}^{(\kappa)}$ inside $\mathcal{D}$. Every interior vertex is shared by exactly six elementary triangles. The number of vertices with multiplicity 2 and 3 in each $G^{(\kappa)}_{\mu, d}$ is included in Table 1. 
\end{lem}
\begin{proof} 
 For $\mu \ne 0$ the segment $G^{(0)}_{-2\mu, d}$  intersects $G^{(0)}_{\mu, d}$ at the tangency point of $G^{(0)}_{\mu, d}$ to $\mathcal{D}$, denoted by $x^{\mathcal{D}}_{\mu}$,
and $G^{(0)}_{\nu, d}$ intersects $G^{(0)}_{\mu, d}$ at the tangency point $x^{\mathcal{D}}_{\nu}$ of $G^{(0)}_{\nu, d}$ to $\mathcal{D}$, where $\nu$ is a solution to the equation $2\nu \equiv -\mu \thickspace ({\rm mod} \thickspace d)$.
 \bigskip\par
1) $d=2q+1, \kappa=0, q \ne 3l+1$. The solution is $\nu=q\mu$. If $\mu \ne 0$ there is no $\nu \ne q\mu, -2\mu$ (mod $d$) such that $\lambda+\mu+\nu \equiv 0$ (mod $d$) with $\lambda,\mu,\nu$ pairwise distinct, therefore $v_{2}(G^{(0)}_{\mu, d})=2, v_{3}(G^{(0)}_{\mu, d})=\frac{d-3}{2} $. When $\mu = 0$ we have $x^{\mathcal{D}}_{\mu}=x^{\mathcal{D}}_{q\mu}$, hence all the vertices in $G^{(0)}_{0, d}$ have multiplicity 3 and $v_{2}(G^{(0)}_{0, d})=0, v_{3}(G^{(0)}_{0, d})=\frac{d-1}{2}$. 
 \bigskip\par
2) $d=2q+1, \kappa=0, q = 3l+1, l=1, 2, 3...$. These cases are similar to the cases discussed in 1) except that $x^{\mathcal{D}}_{\mu}=x^{\mathcal{D}}_{q\mu}$ when $\mu = 0, 2l+1, 4l+2$, and we have $v_{2}(G^{(0)}_{\mu, d})=0, v_{3}(G^{(0)}_{\mu, d})=\frac{d-1}{2}$ if $\mu = 0, 2l+1, 4l+2$ and $v_{2}(G^{(0)}_{\mu, d})=2, v_{3}(G^{(0)}_{\mu, d})=\frac{d-3}{2}$ otherwise.
 \bigskip\par
3) $d=2q, \kappa=0, q = 2l, l=2, 3, 4...$. If $q \ne 6p, p=1, 2, 3...$ then for $\mu \in I_{even}$ there are two different solutions: $\nu=d-\frac{\mu}{2}, \frac{d-\mu}{2}$. Hence $v_{2}(G^{(0)}_{\mu, d})=3, v_{3}(G^{(0)}_{\mu, d})=\frac{d-4}{2}$. If $\mu=0$ there is only one solution which is $\nu=\frac{d}{2}$. For $\mu=1, 3, 5...d-1$ there is no solution distinct from $\mu$ and only $G^{(0)}_{-2\mu, d}$ intersects $G^{(0)}_{\mu, d}$ at the tangency point of $G^{(0)}_{\mu, d}$ to $\mathcal{D}$. Therefore $v_{2}(G^{(0)}_{\mu, d})=1, v_{3}(G^{(0)}_{\mu, d})=\frac{d-2}{2}$ for $ \mu \in I_{odd}\cup \{0\}$.
 \bigskip\par
The remaining cases are treated along the same lines and are given in Table.1. In all cases all vertices of multiplicity 2 are on $\mathcal{D}$ and all interior vertices have multiplicity 3, hence they are shared by exactly six triangles. 
\end{proof}
\par
  \begin{figure}[h]
 \includegraphics[width=20pc]{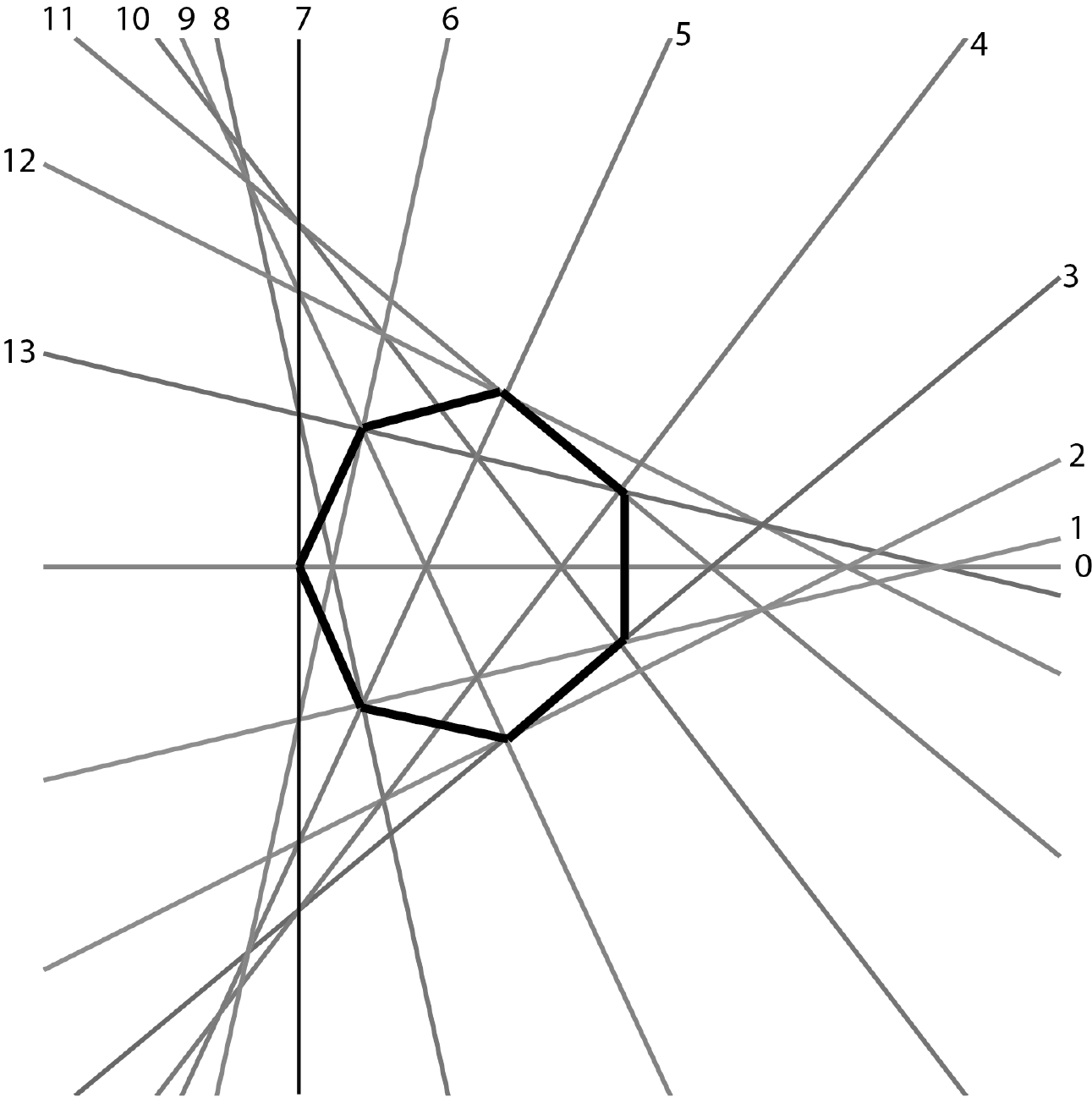}
\caption{\label{label} The configuration of lines $\mathcal{A}^{(0)}_{14}$. The segments $G^{(0)}_{\mu, 14}$, for $\mu=0, 1, ..., 13$, lie on the lines $L_{14, 0, \mu}(x,y)=0$ which are labelled by $\mu$.  The heptagon associated with the random tilings is included.  }
\end{figure}
  \begin{figure}[h]
 \includegraphics[width=20pc]{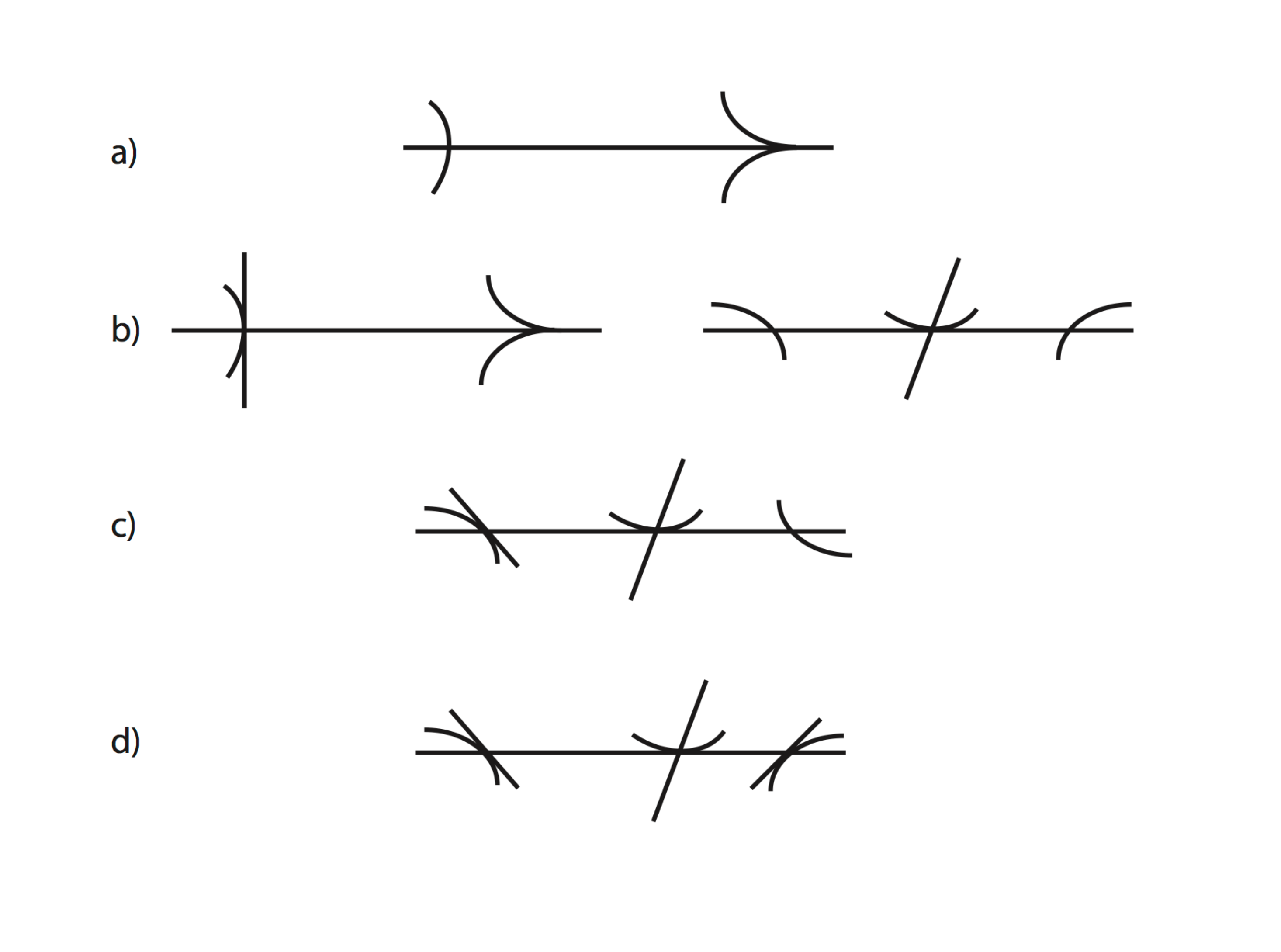}
\caption{\label{label} Schematic representation of the cases with $v_{2}(G^{(\kappa)}_{\mu, d})=0,1,2,3$. The curved fragments represent parts of the deltoid in the neighbourhood of points where the line $L_{d,\kappa,\mu}$ either intersects transversally, passes through one of the deltoid cusps, or is tangent to $\mathcal{D}$.  }
\end{figure}

\par

\begin{lem}
Every segment $G^{(\kappa)}_{\mu,d}$ is cut by the other segments into pieces forming the following sequences: 
\par
1)  $d=2q+1, q \ne 3l+1,  \kappa=0$: $(S_{n})_{n \in I_{odd}}$,
$\forall  \mu \ne 0$, and  
$(S_{n})_{n \in I_{odd}\setminus\{1\}}$ 
for $\mu = 0$.
\par
2)  $d=2q+1, q=3l+1$:
\par
2.1) $\kappa=0$: $(S_{n})_{n \in I_{odd}}$, $ \mu \notin \{ 0,2l+1, 4l+2\}$ and 
 $(S_{n})_{n \in I_{odd}\setminus\{1\}}$  if $ \mu \in \{ 0,2l+1, 4l+2\}$. 
\par
2.2) $\kappa=\mp 2$:  $(S_{n})_{n \in I_{odd}}, \forall \mu$.
\par
3) $d=2q$, $q=3l$:
\par
3.1) $\kappa=0$: $(S_{n})_{n \in I_{odd}}$ if $ \mu $ is even with $ \mu \notin \{ 0,2l, 4l\}$ and  $(S_{n})_{n \in I_{odd}\setminus\{1\}}$  if $ \mu \in \{ 0,2l, 4l\}$. 
\par
3.2) $\kappa=\mp 2$: $(S_{n})_{n \in I_{odd}}$ if  $| \mu |$ is even or 0. 
\par
3.3) $\kappa=0, \mp 2$: $(S_{n})_{n \in I_{even}}$  if $| \mu |$ is odd. 
\par
4) $d=2q$,  $q \ne 3l$: $(S_{n})_{n \in I_{odd}}$ for $\mu$  even, $(S_{n})_{n \in I_{even}}$ for $\mu$  odd, and  $(S_{n})_{n \in I_{odd}\setminus\{1\}}$ for $\mu=0$.
\end{lem}
\begin{proof}
We write $p_{\mu,\nu}$ instead of $p(\frac{\mu\pi}{d},\frac{\nu\pi}{d})$. For $d=2q+1, \kappa=0, q \ne 3l+1$ we take $p_{\mu, q\mu}$ as starting point in $G^{(0)}_{\mu,d}$ for $\mu \ne 0$, which corresponds to the left most point in Fig.2.c. This is a multiplicity 2 point, and we have a number of consecutive multiplicity 3 points $p_{\mu, q\mu+1}=p_{\mu, q\mu-1}, p_{\mu, q\mu+2}=p_{\mu, q\mu-2},... p_{\mu, q\mu+n-1}=p_{\mu, q\mu-n+1}$ until we reach a value of $n$ such that either $\mu q-n=\mu$ or  $\mu q+n=\mu$, namely, either $\mu q+n=-2\mu$ or $\mu q-n=-2\mu$ which corresponds to a point of multiplicity 2 (the second point from the left in Fig.2.c). The remaining vertex points in $G^{(0)}_{\mu,d}$ are $p_{\mu, q\mu+n+1}=p_{\mu, q\mu-n-1}, p_{\mu, q\mu+n+2}=p_{\mu, q\mu-n-2},...p_{\mu, q\mu+q}=p_{\mu, q\mu-q}$ and have multiplicity 3. Now $|p_{\mu,\mu q+l}-p_{\mu,\mu q+l-1}|=4s_{2l-1}s_{1}$ and we get the sequence  $(S_{n})_{n \in I_{odd}}$.
 \par
On $G^{(0)}_{0,d}$ all the points have multiplicity 3 and the sequence of points, starting from the nearest point to the intersection of $G^{(0)}_{0,d}$ with the deltoid cusp (Fig.2a, right) is $p_{0, d-1}=p_{0, 1} , p_{0, d-2}=p_{0, 2} ,...p_{0, d-q}=p_{0, q}$. Having in mind 
$|p_{0,d-n}-p_{0,d-n-1}|=4s_{2n+1}s_{1}, n=1,2,...q-1$, we obtain the sequence $(S_{n})_{n \in I_{odd}\setminus\{1\}}$.
\par
For $d=2q+1, q=3l+1$ the segments with $ \mu \in \{ 0,2l+1, 4l+2\}$ are analogous to the segment with $ \mu=0$ when $q \ne 3l+1$ and those with $ \mu \notin \{ 0,2l+1, 4l+2\}$ are like $\mu \ne 0$. The remaining cases can be analysed in a similar way having in mind the results of Lemma 2.6 summarised in Table 1.
 \end{proof}
\par
 \section{Substitution tilings of the plane}
\bigskip\par
\begin{defn}
A tiling of an $n$-dimensional space $S$ is a decomposition of $S$ into a countable number of $n$-cells $T$, called prototiles, such that $S$ is the union of tiles and distinct tiles have non-intersecting interiors. A collection of tiles, any two of which intersect only in their boundaries, is called a patch.
\end{defn}
 \par
 A substitution rule determines how to replace each prototile by a patch of tiles. Iteration of the substitution rules gives in the limit a substitution tiling of $S$. In this work $S$ is the real plane and the prototiles are triangles (see \cite{gah15, mal15} for recent related work done with computer assistance where the prototiles are triangles and rhombuses respectively).
\par
 In \cite{nis96} the prototile edges are marked with arrows. Prototiles having the same shape but different arrow decorations have different substitution rules. We use a procedure to distinguish prototiles having the same shape that we call interior decoration, or ID for short. If two arrangements $\mathcal{X}_{d}, \mathcal{X}_{2d}$ satisfy $\mathcal{X}_{d}\subset \mathcal{X}_{2d}$ then an elementary triangle $t_{2d}$ in $\mathcal{X}_{2d}$ is inscribed in an elementary triangle $t_{d}$  of $\mathcal{X}_{d}$ which is a scaled copy of $t_{2d}$. We decorate the interior of $t_{d}$ with the inscribed copy $t_{2d}$ and the decorated prototile is $T=t_{d}\cup t_{2d}$ (see \cite{esc17}, Fig.2, p.105). The set of prototiles $T$ is denoted by $\mathcal{F}_{d}$. From the definitions we have $G^{(0)}_{\lambda,d}=G^{(0)}_{2\lambda,2d}, G^{(-2)}_{\lambda,d}=G^{(2)}_{2\lambda+2,2d}, G^{(2)}_{\lambda,d}=G^{(-2)}_{2\lambda-2,2d}$, therefore $\mathcal{A}^{(\kappa)}_{d} \subset \mathcal{A}^{(-\kappa)}_{2d}$.  
   \par
The elementary triangle  $\Delta^{(\kappa)}_{d}( \lambda , \mu, \nu)$ in  $\mathcal{G}_{\Delta,d}^{(\kappa)}$ is congruent to $\Delta^{(\kappa)}_{2d}( 2 \lambda ,2 \mu,2 \nu)$ in $\mathcal{G}_{\Delta,2d}^{(\kappa)}$.  If  $\sigma_{2d} -\kappa= \pm2$ then the elementary triangle of $\mathcal{G}_{\Delta,2d}^{(\kappa)}$ having its vertices on the edges in $\Delta^{(\kappa)}_{2d}( 2 \lambda ,2 \mu,2 \nu)$ is $\Delta^{(\kappa)}_{2d}( 2 \lambda \mp 1 ,2 \mu \mp 1,2 \nu \mp 1)$, which on the other hand is a congruent copy of $\Delta^{(\kappa)}_{2d}( 2 \lambda ,2 \mu,2 \nu)$  scaled by $\iota^{-1}_{2d,2}$. Now each edge of $\Delta^{(\kappa)}_{d}( \lambda , \mu, \nu)$ is subdivided into two sections defined by the vertex of the inscribed triangle $\Delta^{(\kappa)}_{2d}( 2 \lambda \mp 1 ,2 \mu \mp 1,2 \nu \mp 1)$ lying on the edge. The sequences of sections on the edges of $\Delta^{(\kappa)}_{2d}( 2 \lambda ,2 \mu,2 \nu)$ lying on the segments and determining the ID are the following: $G^{(\kappa)}_{2\lambda, 2d}: (S_{2\nu-2\mu  \mp 1}, S_{2\nu-2\mu \pm1})$; $G^{(\kappa)}_{2\mu, 2d}: (S_{2\lambda-2\nu  \mp 1+2d}, S_{2\lambda-2\nu \pm 1+2d})$; $G^{(\kappa)}_{2\nu, 2d}:(S_{2\mu-2\lambda  \mp1}, S_{2\mu-2\lambda \pm1})$ for $ 0\leq 2\lambda < 2\mu < 2\nu < 2d; 2 \lambda +2 \mu+2 \nu-\kappa=\pm 2$. 

\begin{lem}
Given integers $\alpha, \beta, \gamma$ with $\alpha+ \beta+ \gamma \equiv 0$ (mod $d$), $\alpha \le \beta \le \gamma $, $\sigma_{d} -\kappa \equiv \pm1$ (mod $d$) and $l=1,2,3...$, we have the following properties: 
\par
1) $\mathcal{G}_{\Delta,d}^{(0)}, d=2q+1, q \ne 3l+1,$ or $d=2q, q \ne 3l$. The number of triangles in $\mathcal{G}_{\Delta,d}^{(0)}$  with angles $\frac{\alpha\pi}{d}, \frac{\beta\pi}{d}, \frac{\gamma\pi}{d}$ and the given $\sigma_{d} $ equals one, if $\alpha, \beta, \gamma$ are not pairwise distinct (the triangle then is isosceles), and equals two otherwise, and the two triangles then differ with respect to the ID. 
\par
2) $\mathcal{G}_{\Delta,d}^{(0)}, d=3q$. There are no isosceles elementary triangles. The non isosceles elementary triangles appear provided $\beta-\alpha$ is not a multiple of 3. There are three congruent copies of each elementary triangle but they have the same ID and must be considered as a single prototile. 
\par
3)  $\mathcal{G}_{\Delta,d}^{(0)}\cup \mathcal{G}_{\Delta,d}^{(-2)}\cup \mathcal{G}_{\Delta,d}^{(2)},d=3q$. There are isosceles elementary triangles and one of them is equilateral. There are two types of non isosceles elementary triangles. The number of  non isosceles elementary triangles with angles $\frac{\alpha\pi}{d}, \frac{\beta\pi}{d}, \frac{\gamma\pi}{d}$ and a given $\sigma_{d}$ equals one, if  $\beta-\alpha$ is not a multiple of 3 and  $|\sigma_{d}|=1$, and equals two if  $\beta-\alpha$ is a multiple of 3 and  $|\sigma_{d}|=3$, and the two triangles then differ with respect to the ID.
\end{lem}
\begin{proof}
\par
1.1) $\mathcal{G}_{\Delta,d}^{(0)}, d=2q+1, q \ne 3l+1$. Any triangle with angles $\frac{\alpha\pi}{d}, \frac{\beta\pi}{d}, \frac{\gamma\pi}{d}$ is either $\Delta_{d}^{(\kappa)}(\mu- \beta,\mu,\mu+\alpha)$ or $\Delta_{d}^{(\kappa)}(\mu- \alpha,\mu,\mu+\beta)$. In this case $d  \not\equiv 0$ (mod 3), therefore by Prop. 2.2 there is no equilateral triangle and we may assume $\gamma \ne \alpha, \beta$.  We look for solutions to the congruences $\mu \mp \beta+\mu+\mu \pm \alpha \equiv \sigma_{d} $ (mod $d$). 
\par
If $\alpha=\beta$ (isosceles triangle) then $3\mu \equiv \sigma_{d} $ (mod $d$). If an elementary isosceles triangle has sides $S_{a},S_{b},S_{b}$ we now prove that the side $S_{a}$ can be in $G^{(0)}_{\mu, d}$ for only two different values of $\mu$. We have to look for solutions of $3 \mu \equiv \pm 1$ (mod $d$) and, according to Prop. 2.2 there is exactly one solution in each case. The unique solution to $3 \mu \equiv  1$ (mod $d$) is $\mu=2l$ for $d=6l-1$ and $\mu=4l+1$ for $d=6l+1$. The solution to $3 \mu \equiv  -1$ (mod $d$) is $\mu=4l-1$ for $d=6l-1$ and $\mu=2l$ for $d=6l+1$. By Lemma 2.6 the number of subdivisions of $G^{(0)}_{\mu, d}$ in these cases is $q$, and we have $q$ different isosceles triangles: $\Delta_{d}^{(0)}(\mu- \alpha,\mu,\mu+\alpha), \alpha=1,2,...q$. The triangle $\Delta_{d}^{(0)}(-2l- \alpha,-2l,-2l+\alpha)$ in $G^{(0)}_{-2l, d}$ is a congruent copy of $\Delta_{d}^{(0)}(2l- \alpha,2l,2l+\alpha)$ in $G^{(0)}_{2l, d}$ rotated by $\frac{(2l \mp 1)\pi}{d}$ for $d=6l \mp 1$. But it differs on the ID because they have values of $\sigma_{d} $ with opposite signs: the sequences of sections on the edges are reversed. 
\par
If $\alpha \ne \beta$, the equation $3x\equiv b$ (mod $d$) has one solution because $gcd(3,d)=1$ therefore we have one solution for $3 \mu_{1}\equiv \sigma_{d}  - \alpha + \beta$  (mod $d$) and another one for $3 \mu_{2}\equiv \sigma_{d}  + \alpha - \beta$  (mod $d$).
For $d=6l+1$, given a solution of   $3 \mu_{1}\equiv \sigma_{d}  - \alpha + \beta$  (mod $d$) with $\sigma_{d}=-1$ corresponding to $\Delta_{d}^{(0)}(\mu_{1}- \beta,\mu_{1},\mu_{1}+\alpha)$ then $\bar{\mu}_{1}=\mu_{1}+2l+1$  is a solution of $3 \bar{\mu}_{1}\equiv 1  - \alpha + \beta$  (mod $d$) and $\Delta_{d}^{(0)}(\mu_{1}- \beta+2l+1,\mu_{1}+2l+1,\mu_{1}+\alpha+2l+1)$ is a congruent copy of $\Delta_{d}^{(0)}(\mu_{1}- \beta,\mu_{1},\mu_{1}+\alpha)$ but with different ID. Also if $\Delta_{d}^{(\kappa)}(\mu_{2}- \alpha,\mu_{2},\mu_{2}+\beta)$ corresponds to the solution of $3 \mu_{2}\equiv \sigma_{d}  + \alpha - \beta$  (mod $d$) with $\sigma_{d}=-1$ then we have its congruent triangle $\Delta_{d}^{(\kappa)}(\mu_{2}- \alpha+2l+1,\mu_{2}+2l+1,\mu_{2}+\beta+2l+1)$ with different ID. For $d=6l-1$ if $\Delta_{d}^{(0)}(\mu_{1}- \beta,\mu_{1},\mu_{1}+\alpha)$ is linked with the solution of $3 \mu_{1}\equiv \sigma_{d}  - \alpha + \beta$  (mod $d$) with $\sigma_{d}=1$, then we have $\Delta_{d}^{(0)}(\mu_{1}- \beta+2l-1,\mu_{1}+2l-1,\mu_{1}+\alpha+2l-1)$ and if $\Delta_{d}^{(\kappa)}(\mu_{2}- \alpha,\mu_{2},\mu_{2}+\beta)$ corresponds to the solution of $3 \mu_{2}\equiv \sigma_{d}  + \alpha - \beta$  (mod $d$) with $\sigma_{d}=1$, then its congruent triangle is $\Delta_{d}^{(\kappa)}(\mu_{2}- \alpha+2l-1,\mu_{2}+2l-1,\mu_{2}+\beta+2l-1)$.
\par
 If $p(d,n)$ denotes the number of partitions of $d$ into exactly $n$ positive integers then $p(d-{n \choose 2},n)$ is the number of ways of writing $d$ as a sum of $n$ different positive integers  \cite{com74}. The number of elementary non isosceles triangles with $\sigma_{d}=1$ or $\sigma_{d}=-1$ is equal to  $ 2 p(d- {3 \choose 2},3)=2\lfloor \frac{(d-3)^2}{12} \rceil $,  where $ \lfloor $ $\rceil $ is the nearest integer function. The total number of elementary triangles in $\mathcal{G}_{\Delta,d}^{(0)}$ is therefore $d-1+4 \lfloor \frac{(d-3)^2}{12} \rceil $ for $d=6l \pm 1$.
 \bigskip\par
1.2) $\mathcal{G}_{\Delta,d}^{(0)}, d=2q, q \ne 3l$. This case is analogous to 1.1). If $\alpha=\beta$ then the unique solution to $3 \mu \equiv  1$ (mod $d$) is $\mu=2l+1$ for $d=6l+2$ and $\mu=4l+3$ for $d=6l+4$. The solution to $3 \mu \equiv  -1$ (mod $d$) is $\mu=4l+1$ for $d=6l+2$ and $\mu=2l+1$ for $d=6l+4$. By Lemma 2.6 the number of subdivisions of $G^{(0)}_{\mu, d}$ in these cases is $q-1$, and we have $q-1$ different isosceles triangles: $\Delta_{d}^{(0)}(\mu- \alpha,\mu,\mu+\alpha), \alpha=1,2,...q-1$. 
\par
If $\alpha \ne \beta$, the equation $3x\equiv b$ (mod $d$) has one solution.
For $d=6l+2$, given a solution of   $3 \mu_{1}\equiv \sigma_{d}  - \alpha + \beta$  (mod $d$) with $\sigma_{d}=-1$ corresponding to $\Delta_{d}^{(0)}(\mu_{1}- \beta,\mu_{1},\mu_{1}+\alpha)$ then $\bar{\mu}_{1}=\mu_{1}-2l$  is a solution of $3 \bar{\mu}_{1}\equiv 1  - \alpha + \beta$  (mod $d$) and $\Delta_{d}^{(0)}(\mu_{1}- \beta-2l,\mu_{1}-2l,\mu_{1}+\alpha-2l)$ is a congruent copy of $\Delta_{d}^{(0)}(\mu_{1}- \beta,\mu_{1},\mu_{1}+\alpha)$ but with different ID. Also if $\Delta_{d}^{(\kappa)}(\mu_{2}- \alpha,\mu_{2},\mu_{2}+\beta)$ corresponds to the solution of $3 \mu_{2}\equiv \sigma_{d}  + \alpha - \beta$  (mod $d$) with $\sigma_{d}=-1$ we have its congruent triangle $\Delta_{d}^{(\kappa)}(\mu_{2}- \alpha+2l,\mu_{2}+2l,\mu_{2}+\beta+2l)$ with different ID. For $d=6l+4$ if $\Delta_{d}^{(0)}(\mu_{1}- \beta,\mu_{1},\mu_{1}+\alpha)$ is related to the solution of $3 \mu_{1}\equiv \sigma_{d}  - \alpha + \beta$  (mod $d$) with $\sigma_{d}=-1$, then we have $\Delta_{d}^{(0)}(\mu_{1}- \beta+2l+2,\mu_{1}+2l+2,\mu_{1}+\alpha+2l+2)$ and if $\Delta_{d}^{(\kappa)}(\mu_{2}- \alpha,\mu_{2},\mu_{2}+\beta)$ is associated to the solution of $3 \mu_{2}\equiv \sigma_{d}  + \alpha - \beta$  (mod $d$) with $\sigma_{d}=-1$ its congruent triangle is $\Delta_{d}^{(\kappa)}(\mu_{2}- \alpha-2l-2,\mu_{2}-2l-2,\mu_{2}+\beta-2l-2)$.
\par
 The total number of elementary triangles in $\mathcal{G}_{\Delta,d}^{(0)}$ is therefore $d-2+4 \lfloor \frac{(d-3)^2}{12} \rceil $.   
 \bigskip\par
2) $\mathcal{G}_{\Delta,d}^{(0)}, d=3q$. Now $gcd(3,d)=3$ then $3 \mu \equiv  \pm 1$ (mod $d$) has no solution and we do not have elementary isosceles triangles. For $\alpha \ne \beta$ if $\sigma_{d}= \pm 1$ then $3 \mu_{1}\equiv \sigma_{d}  - \alpha + \beta$  (mod $d$) has solution, according to Prop.2.2, when $\alpha-\beta \in A=\{-\sigma_{d} 	\pm 3n\}_{n=0,1,2,...}$ and  $3 \mu_{2}\equiv \sigma_{d}  + \alpha - \beta$  (mod $d$) has solution when $\alpha-\beta \in B=\{\sigma_{d} 	\pm 3n\}_{n=0,1,2,...}$. In both cases we have 3 solutions when $\beta-\alpha$ is not divisible by 3: $\Delta_{d}^{(0)}(\lambda+nq,\mu+nq,\nu+nq), n=0,1,2$ which correspond to the same prototile (they have the same ID) rotated by $\frac{2\pi n}{3}$.  
\par
For $d=3q$, the number of solutions of $\beta-\alpha=3m, m=1,2,3,...$ is $q-2m-1$, therefore the total number of elementary triangles in $\mathcal{G}_{\Delta,3q}^{(0)}$ is $6 (\lfloor \frac{(d-3)^2}{12} \rceil -l(l-1))$ for $q=2l+1$ and $6 (\lfloor \frac{(d-3)^2}{12} \rceil -(l-1)^2)$ for $q=2l$.
 \bigskip\par
3) $\mathcal{G}_{\Delta,d}^{(0)}\cup \mathcal{G}_{\Delta,d}^{(-2)}\cup \mathcal{G}_{\Delta,d}^{(2)},d=3q$. The case $\kappa=0$ has already been analysed. For $\kappa=-2$ the elementary triangles have $\sigma_{d}\in \{-3,-1\}$. If $\alpha=\beta$ then $3 \mu \equiv -1$ (mod $d$) has no solution. If $q=2l$, then $\mu=2l-1+n, n=0, q, 2q$ are the solutions to $3 \mu \equiv  -3$ (mod $d$) and have the same ID. According to Lemma 2.6 when $\mu \in I_{odd}$, the number of subdivisions of the segment  $G^{(-2)}_{\mu, d}$  is $3l-1$. The triangles with an edge on $G^{(-2)}_{\mu, d}$ are $\Delta_{d}^{(-2)}(\mu- \alpha,\mu,\mu+\alpha), \alpha=1,2,...3l-1$ and the equilateral, which corresponds to $\alpha=2l$, has its sides on  $G^{(-2)}_{\mu, d}, \mu=2l-1+n, n=0, q, 2q$, therefore the number of isosceles triangles is $1+ 3 (3l-1-1)=9l-5$. For $q=2l+1$, the solutions to $3 \mu \equiv  -3$ (mod $d$) are $\mu=2l+n, n=0, q, 2q$, the segment $G^{(-2)}_{\mu, d}$ is subdivided into $3l+1$ segments and the number of isosceles triangles is $9l+1$. The case $\kappa=2$ can be studied in a similar way having in mind that $\sigma_{d}\in \{1,3\}$.
\par
For $\alpha \ne \beta$ and $\kappa=-2$, the equation $3 \mu_{1}\equiv \sigma_{d}  - \alpha + \beta$ (mod $d$), with $\sigma_{d}=-1$ has solution if $\beta-\alpha \in \{3m+1\}_{m=1,2,3,...}$ and $3 \mu_{2}\equiv \sigma_{d}  + \alpha - \beta$ (mod $d$) when $\beta-\alpha \in \{3m-1\}_{m=1,2,3,...}$. In both cases we have 3 solutions which correspond to the same prototile rotated by $\frac{2\pi n}{3}, n=0,1,2$. For  $\sigma_{d}=-3$ the equations $3 \mu_{1}\equiv \sigma_{d}  - \alpha + \beta$ (mod $d$) have solutions when $\beta-\alpha \in \{3m\}_{m=1,2,3,...}$. But now for each solution of $\mu_{1}$ linked to $\Delta_{d}^{(-2)}(\mu_{1}- \alpha,\mu_{1},\mu_{1}+\alpha)$ we have a solution of $3 \bar{\mu}_{1}\equiv 3  - \alpha + \beta$ (mod $d$) in $G^{(2)}_{\mu, d}$ corresponding to a congruent triangle $\Delta_{d}^{(2)}(\bar{\mu}_{1}- \alpha,\bar{\mu}_{1},\bar{\mu}_{1}+\alpha)$ having a different ID, and the same occurs for $\mu_{2}$. The number of non isosceles elementary triangles with two different ID  is, up to mirror reflection, $3l(l-1)$ for $q=2l+1$ and $3(l-1)^2$ for $q=2l$.
\par
The total number of elementary triangles in $\mathcal{G}_{\Delta,3q}^{(-2)}$ or in $\mathcal{G}_{\Delta,3q}^{(2)}$ is $9l+1+3 (\lfloor \frac{(d-3)^2}{12} \rceil -l(l-1))+6l(l-1)$ for $q=2l+1$ and $9l-5+3(\lfloor \frac{(d-3)^2}{12} \rceil -(l-1)^2)+6(l-1)^2$ for $q=2l$.
\end{proof}
    \par
  \begin{figure}[h]
 \includegraphics[width=21pc]{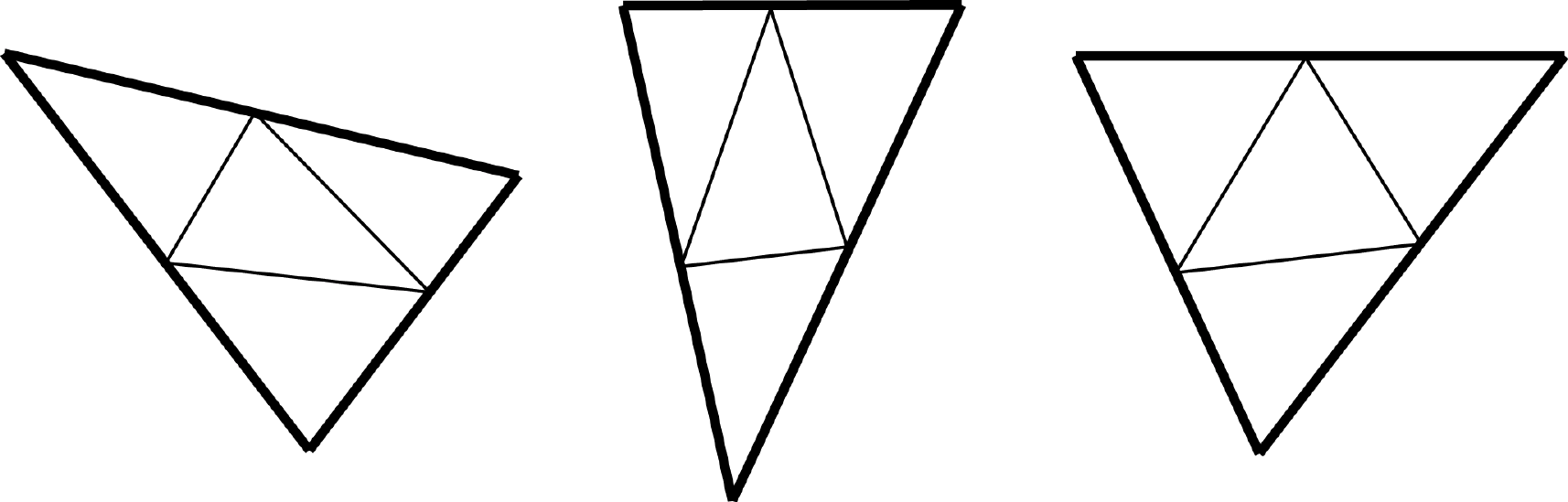}
\caption{\label{label} The elementary triangles $F:= \Delta_{14}^{(0)}(4,10,13),  \hat{F}:= \Delta_{14}^{(0)}(0,5,8)$ and $G:= \Delta_{14}^{(0)}(0,4,9)$ (from left to right) with the $ID$. $F$ and  $\hat{F}$ have the same edge subdivision.}
\end{figure}
   \bigskip\par
{ \it Example} . We denote by $ \hat{X}$ an elementary triangle which is the mirror image of the non-isosceles $X$ but with edges having the same subdivision (Fig.3). For $d=14$ the elementary triangles with $\sigma_{14}=-1$ are the following: $$A:= \Delta_{14}^{(0)}(0,1,12), \hat{A}:= \Delta_{14}^{(0)}(3,4,6), B:= \Delta_{14}^{(0)}(2,12,13),  \hat{B}:= \Delta_{14}^{(0)}(2,5,6), C:= \Delta_{14}^{(0)}(0,2,11),$$ $$ \hat{C}:= \Delta_{14}^{(0)}(2,4,7), D:= \Delta_{14}^{(0)}(3,11,13),  \hat{D}:=\Delta_{14}^{(0)}(1,5,7), E:= \Delta_{14}^{(0)}(0,3,10),  \hat{E}:= \Delta_{14}^{(0)}(1,4,8),$$ $$  F:= \Delta_{14}^{(0)}(4,10,13),  \hat{F}:= \Delta_{14}^{(0)}(0,5,8), G:= \Delta_{14}^{(0)}(0,4,9), H:= \Delta_{14}^{(0)}(5,9,13), I:= \Delta_{14}^{(0)}(6,8,13),$$ $$   \hat{I}:= \Delta_{14}^{(0)}(5,10,12), J:= \Delta_{14}^{(0)}(0,6,7),  \hat{J}:= \Delta_{14}^{(0)}(4,11,12), K:= \Delta_{14}^{(0)}(1,2,10),  \hat{K}:= \Delta_{14}^{(0)}(2,3,8),$$ $$  L:= \Delta_{14}^{(0)}(1,3,9), M:= \Delta_{14}^{(0)}(6,9,12), N:= \Delta_{14}^{(0)}(7,8,12),  \hat{N}:= \Delta_{14}^{(0)}(6,10,11), O:= \Delta_{14}^{(0)}(7,9,11), $$ $$ P:= \Delta_{14}^{(0)}(8,9,10)$$

\par
\begin{lem}
The edges of the inflated triangles with inflation factor $\iota_{d,l}$ are subdivided according to the sequences  $(S_{l-j+1},S_{l-j+3},...,S_{l+j-1})$ if $1\le j\le l$ and $(S_{j-l+1},S_{j-l+3},...,S_{j+l-1})$ if $l+1\le j\le q= \lceil \frac{d-1}{2} \rceil $.
\end{lem}
\begin{proof}
We first prove the trigonometric identities \cite{nis96,fre98, esc08}
     \begin{equation}
\frac{s_{l}}{s_{1}}s_{j}=\sum^{j-1}_{k=0}s_{l-j+2k+1}, 1\le j\le l; 
 \end{equation}
     \begin{equation}
 \frac{s_{l}}{s_{1}}s_{j}=\sum^{l-1}_{k=0}s_{j-l+2k+1}, l+1\le j\le q
 \end{equation}

Eq. (3.1) can be obtained by induction on $j$. For $j=1$ it is an identity. For $j=2$ we have to prove 
      \begin{equation}
s_{l}s_{2}=s_{1}(s_{l-1}+s_{l+1})
\end{equation}
 which is a consequence of the Mollweide\'{}s formula and the law of sines. Now assuming it is true for $j-2$ then it is true for $j$ because $\frac{s_{l}}{s_{1}}(s_{j}-s_{j-2})=s_{l+j-1}+s_{l-j+1}$. On the other hand, for $j=1,2,...q$, the identity
$\frac{s_{l}}{s_{1}}s_{j}=\sum^{l-1}_{k=0}s_{j-l+2k+1} $
can be obtained also by induction on $j$.  After cancelling terms like $s_{n}+s_{-n}$ we get it for $j=1$. For $j=2$ it is given by eq.(3.1). Now assuming it is true for $j-2$ then it is true for $j$ because $\frac{s_{l}}{s_{1}}(s_{j}-s_{j-2})=s_{j+l-1}-s_{j-l-1}$.
\par
Therefore we have
    \begin{equation}
    \iota_{d,l}\cdot 4s_{1}s_{j}=4s_{1}(s_{l-j+1}+s_{l-j+3}+...s_{l+j-1}), 1\le j\le l
     \end{equation}
     and
       \begin{equation}
      \iota_{d,l}\cdot 4s_{1}s_{j}=4s_{1}(s_{j-l+1}+s_{j-l+3}+...s_{j+l-1}),  l+1\le j\le q
     \end{equation}
We get a sequence $(S_{l-j+1},S_{l-j+3},...,S_{l+j-1})$ if $1\le j\le l$ and $(S_{j-l+1},S_{j-l+3},...,S_{j+l-1})$ if $l+1\le j\le q$, and both represent a unique  interval  on every $G_{\mu}$. 
\end{proof}
\par

\begin{lem}
If two triangles with the same ID and with $\sigma_{d} -\kappa \not\equiv \pm1$ (mod $d$) are congruent  they are dissected in the same way.
\end{lem}
\begin{proof}
\par
1.1) $\mathcal{G}_{\Delta,d}^{(0)}, d=2q+1, q \ne 3l+1, l=1,2,3...$. If $\alpha=\beta$ then the solution to $3\mu \equiv \sigma_{d}$ (mod $d$), $\sigma_{d}=2,3,...q$ is $\mu=2l \sigma_{d}$ for $d=6l-1$ and $\mu=(4l+1) \sigma_{d}$ for $d=6l+1$. If $\alpha \ne \beta$, given a solution to  $3 \mu_{1}\equiv \sigma_{d}  - \alpha + \beta$  (mod $d$) then $\bar{\mu}_{1}=\mu_{1}-(2l+1)\sigma_{d}$  is a solution of $3 \bar{\mu}_{1}\equiv -\sigma_{d}  - \alpha + \beta$  (mod $d$) for $d=6l+1$ and $\bar{\mu}_{1}=\mu_{1}+(2l-1)\sigma_{d}$ for $d=6l-1$ . We have the same results by replacing $\bar{\mu}_{1}$ by $\bar{\mu}_{2}$ if we analyse $3 \bar{\mu}_{2}\equiv \sigma_{d}  + \alpha - \beta$.  
 \bigskip\par
1.2) $\mathcal{G}_{\Delta,d}^{(0)}, d=2q, q \ne 3l, l=1,2,3...$. If $\alpha=\beta$ then the solution to $3\mu \equiv \sigma_{d}$ (mod $d$), $\sigma_{d}=2,3,...q$ is $\mu=(2l+1) \sigma_{d}$ for $d=6l+2$ and $\mu=(4l+3) \sigma_{d}$ for $d=6l+4$. If $\alpha \ne \beta$, given a solution to  $3 \mu_{1}\equiv \sigma_{d}  - \alpha + \beta$  (mod $d$) then $\bar{\mu}_{1}=\mu_{1}+2l\sigma_{d}$  is a solution of $3 \bar{\mu}_{1}\equiv -\sigma_{d}  - \alpha + \beta$  (mod $d$) for $d=6l+2$ and $\bar{\mu}_{1}=\mu_{1}-(2l+2)\sigma_{d}$ for $d=6l+4$. We have the same results if we analyse $3 \bar{\mu}_{2}\equiv \sigma_{d}  + \alpha - \beta$ (mod $d$).  
\par
When $d$ is not divisible by 3,  $\lambda+ \mu+\nu= -1$ and $n$ is a solution to $3n \equiv p+1$ (mod $d$)  then
       \begin{equation}
\iota_{d,p}\cdot \Delta_{d}^{(0)}( \pm\lambda ,\pm \mu, \pm\nu) \cong \Delta_{d}^{(0)}(\pm \lambda\pm n,\pm \mu \pm n,\pm \nu\pm n)
     \end{equation} 
 \bigskip\par
2) $\mathcal{G}_{\Delta,d}^{(0)}, d=3q, q = 2l+1$ or $q = 2l$, $l=1,2,3...$. Now there are inflated elementary triangles in $\mathcal{G}_{\Delta,d}^{(0)}$ only when  $\sigma_{d} $ is not a multiple of 3. If $\sigma_{d} >1$ the only values giving solutions are $\sigma_{d} \in \{ \pm 1+ 3n\}_{n=1,2,...,l}$ and we will not have tilings with inflation factors $\iota_{d,p}$ with $p=\sigma_{d}-\kappa=\sigma_{d}$ a multiple of 3 and prototiles contained in $\mathcal{G}_{\Delta,d}^{(0)}$. However there are non elementary triangles which are not inflated copies of the elementary ones. For instance the non-elementary  triangle $\Delta_{9}^{(0)}(2,5,8)$ has $\sigma_{d} =6$ and it is one of the three inflated copies (in this case the inflation factor is $\iota_{9,3}$) of the elementary equilateral triangle appearing in $\mathcal{G}_{\Delta,9}^{(\pm 2)} $ (the other two appear in $\mathcal{G}_{\Delta,9}^{(\pm 2)} $, see 3) in this Lemma). The values  $\sigma_{d} =\pm 1+3n$ ($n=1$ for $d=9$) correspond to the inflated copies of the elementary triangles appearing in $\mathcal{G}_{\Delta,9}^{(0)} $. 
 \bigskip\par
3)  $\mathcal{G}_{\Delta,d}^{(0)}\cup \mathcal{G}_{\Delta,d}^{(-2)}\cup \mathcal{G}_{\Delta,d}^{(2)},d=3q$. 
For $\kappa=-2$ and $\alpha = \beta$ we have 3 solutions to $3\mu \equiv \sigma_{d}$ (mod $d$) when $ \sigma_{d} \in \{ -3+3n\}_{n \in \{0,1,2,...q-1\}}$. The non elementary  triangles $\Delta_{d}^{(-2)}(n-1-\alpha,n-1,n-1+\alpha), n>0, \alpha=1,2,...m$, with $m=3l-1,3l+1$ for $q=2l,2l+1$ respectively, correspond to inflation factors $\iota_{d,p}$ with $p=\sigma_{d}-\kappa=\sigma_{d}+2$ not a multiple of 3. The remaining two solutions for $\mu$ are $n-1+tq, t=1,2$ and give congruent copies rotated by $\frac{2\pi}{3}$, except for the equilateral triangle.  The inflated triangles with $\iota_{d,p}, p=3n, n=1,2,..., l$ are in $\mathcal{G}_{\Delta,d}^{(0)}$, because we have solutions to $3\mu \equiv \sigma_{d}=d-3n$ (mod $d$) like $\Delta_{d}^{(0)}(-n,q-n,2q-n)$. For $\alpha \ne \beta$ the distribution of inflated triangles in $\mathcal{G}_{\Delta,d}^{(\kappa)}$ can be analysed in a similar way. We have seen in Lemma 3.3 that there are two types of non-isosceles elementary triangles that we call $\Delta_{N1}, \Delta_{N2}$ such that $\Delta_{N2}$ is contained only in $\mathcal{G}_{\Delta,d}^{(-2)}\cup \mathcal{G}_{\Delta,d}^{(2)}$ and there are two $\Delta_{N2}$ differing only on their ID. The inflated copy of any $\Delta_{N1}$ is contained in $\mathcal{G}_{\Delta,d}^{(0)}$ when the inflation factor is $\iota_{d,p}$ with $p$ not a multiple of 3 (see 2) in this Lemma), and it is contained in $\mathcal{G}_{\Delta,d}^{(-2)}\cup \mathcal{G}_{\Delta,d}^{(2)}$ otherwise. On the other hand the inflated copy of any $\Delta_{N2}$ is contained in $\mathcal{G}_{\Delta,d}^{(0)}$ when the inflation factor is $\iota_{d,p}$ with $p$ a multiple of 3 and in $\mathcal{G}_{\Delta,d}^{(-2)}\cup \mathcal{G}_{\Delta,d}^{(2)}$ otherwise.
\par
 For  $d=3q$ and $ 1\leq p<q$ with $p$ non divisible by 3, there is a solution $n$ to the congruences $3n \equiv p+1$ (mod $d$),  $3n \equiv p-1$ (mod $d$) and if $\lambda+ \mu+\nu= \kappa \pm1$ then 
        \begin{equation}
\iota_{d,p}\cdot \Delta_{d}^{(\kappa)}( \lambda , \mu, \nu)  \cong  \Delta_{d}^{(\kappa)}( \lambda+n, \mu+n, \nu+n)
      \end{equation} 
\par\noindent
If $p$ is divisible by 3, then there is no solution to $3n \equiv p \pm 1$ (mod $d$). For $\lambda + \mu+ \nu=-3$ we have:
  \par
        \begin{equation}
\iota_{d,3}\cdot \Delta_{d}^{(-2)}( \lambda , \mu, \nu)  \cong \Delta_{d}^{(0)}( \lambda, \mu, \nu), \forall q
      \end{equation} 
\par\noindent
and
        \begin{equation}
\iota_{d,3m}\cdot \Delta_{d}^{(-2)}( \lambda, \mu, \nu)  \cong \Delta_{d}^{(0)}( \lambda+m+1, \mu+m+1, \nu+m+1), m=2, 3, ..., \lfloor \frac{q}{2}\rfloor, q>3
      \end{equation} 
      For $\lambda + \mu+ \nu=1$ and $d=3q, q=2,3$:
              \begin{equation}
\iota_{d,3}\cdot \Delta_{d}^{(0)}( \lambda , \mu, \nu)  \cong \Delta_{d}^{(-2)}( \lambda+q-2, \mu+q-2, \nu+q-2)
      \end{equation}
      and for $d=3q, q>3$
                    \begin{equation}
\iota_{d,3m}\cdot \Delta_{d}^{(0)}( \lambda , \mu, \nu)  \cong \Delta_{d}^{(-2)}( \lambda+m-1, \mu+m-1, \nu+m-1), m=1,2,..., \lfloor \frac{q}{2}\rfloor
      \end{equation}
      
      The corresponding inflated copies of $\widetilde{\Delta}=\Delta_{d}^{(0)}( -\lambda, -\mu, -\nu)$ appear in $\mathcal{G}_{\Delta,d}^{(2)}$.
      \par
Both $\Delta=\Delta_{2d}^{(0)}( 2\lambda , 2\mu, 2\nu),  2\lambda+ 2\mu+ 2\nu=2 $ and $\widetilde{\Delta}=\Delta_{2d}^{(0)}( -2\lambda, -2\mu, -2\nu)$ are included in $\mathcal{G}_{\Delta,2d}^{(0)}$ and are mirror images with respect to $G^{(0)}_{0, 2d}$. When $\kappa \ne 0$ if $\Delta=\Delta_{2d}^{(\kappa)}( 2\lambda , 2\mu, 2\nu)\subset\mathcal{G}_{\Delta,2d}^{(\kappa)}, 2\lambda+ 2\mu+ 2\nu-\kappa=2$ then $\widetilde{\Delta}=\Delta_{2d}^{(-\kappa)}( -2\lambda, -2\mu, -2\nu)\subset\mathcal{G}_{\Delta,2d}^{(-\kappa)}$. There are exactly two triangles with the same $\sigma_{d}$ and the same ID for all non-isosceles triangles in $\mathcal{G}_{\Delta,d}^{(0)}$ for $d$ non-divisible by 3 and one for $d$ divisible by 3. Every triangle can be reflected in $G_{0,d}^{0}$, they have to be mirror images of each other and hence are dissected in the same way.  For the tilings with prototiles in $\mathcal{G}_{\Delta,d}^{(0)}\cup \mathcal{G}_{\Delta,d}^{(-2)}\cup \mathcal{G}_{\Delta,d}^{(2)}$ the triangles contained in $\mathcal{G}_{\Delta,d}^{(\kappa)}$ appear reflected in $\mathcal{G}_{\Delta,d}^{(-\kappa)}$ and are also dissected in the same way. 

\end{proof}
\par

An edge with length $4 {\rm sin}(\frac{\pi}{d}) {\rm sin}(\frac{\nu\pi}{d})$ and subdivision $(S_{a\mp 1}, S_{a \pm 1})$  induced by the ID of a prototile (we choose for instance anticlockwise orientation) can be represented by a letter $W_{\nu}^{1}$ or $W_{\nu}^{-1}$ provided that $S_{a\mp 1}<S_{a \pm 1}$ or $S_{a\mp 1}>S_{a \pm 1}$ respectively, and by $W_{\nu}^{0}$ if $S_{a-1}=S_{a +1}$. We can get face-to-face tilings with the property that the vertices of the inscribed triangles on the prototiles edges match.  We have seen that if there is an elementary triangle $\Delta:=\Delta(W_{{\nu}_{1}}^{i},W_{{\nu}_{2}}^{j},W_{{\nu}_{3}}^{k}),  i,j,k\in \{-1,0,1\}$ then there exists  $\widetilde{\Delta}:=\Delta(W_{{\nu}_{1}}^{-i},W_{{\nu}_{2}}^{-j},W_{{\nu}_{3}}^{-k})$. 
The mirror image of a word $w=W_{{\nu}_{1}}^{i}W_{{\nu}_{2}}^{j}...W_{{\nu}_{n}}^{k}$ is denoted by $Mir(w)$, the projection of $W^{i}$ into $W$ by $P(W^{i})$, and $\rho$ is the map $\rho(W_{{\nu}_{1}}^{i}W_{{\nu}_{2}}^{j}...W_{{\nu}_{n}}^{k})=W_{{\nu}_{1}}^{-i}W_{{\nu}_{2}}^{-j}...W_{{\nu}_{n}}^{-k}$. We choose the tile dissections in such a way that the inflation rules for the edges of $\Delta$ and $\widetilde{\Delta}$ satisfy 
       \begin{equation}
\phi(W^{i})=Mir(\rho(\phi(W^{-i})))
     \end{equation}
 The common edges of two adjacent triangles are represented by $W^{i},W^{-i}$. Therefore the fact that the tilings are face-to-face is equivalent to $P(\phi^{n}(W^{i}))=Mir(P(\phi^{n}(W^{-i})))$, which is a consequence of Eq.(3.12).
   \bigskip\par
     { \it Example}. The prototiles $F, \hat{F}$ in Fig.3 have edges $W_{3}^{-1}, W_{5}^{-1}, W_{6}^{-1}$ corresponding to their subdivisions $(S_{7}, S_{5})$, $(S_{11}, S_{9})$, $(S_{13}, S_{11})$ respectively. The edges of $G$ are $W_{4}^{-1}, W_{5}^{-1}, W_{5}^{-1}$ with subdivisions $(S_{9}, S_{7})$, $(S_{11}, S_{9})$, $(S_{11}, S_{9})$. A solution to $3n \equiv p+1$ (mod $14$) for $p=3$ is $n=6$. An edge inflation rule corresponding to $\iota_{14,3}$ is $$\phi_{+}(W_{1}^{1})= W_{3}^{-1}, \phi_{+}(W_{2}^{1})= W_{4}^{-1}W_{2}^{-1}, \phi_{+}(W_{3}^{1})= W_{5}^{-1}W_{3}^{-1}W_{1}^{-1}, \phi_{+}(W_{4}^{1})= W_{6}^{-1}W_{4}^{-1}W_{2}^{-1}, $$ $$\phi_{+}(W_{5}^{1})= W_{7}^{0}W_{5}^{-1}W_{3}^{-1}, \phi_{+}(W_{6}^{1})= W_{6}^{1}W_{6}^{-1}W_{4}^{-1}, \phi_{+}(W_{7}^{0})= W_{5}^{1}W_{7}^{0}W_{5}^{-1}$$
    In this case the words $\phi_{+}(W_{{\nu}_{n}}^{k})$ are not palindromic therefore we can define another edge inflation rule by $\phi_{-}(W^{i}):=\phi_{+}(W^{-i})$.
    \par
  \begin{figure}[h]
 \includegraphics[width=25pc]{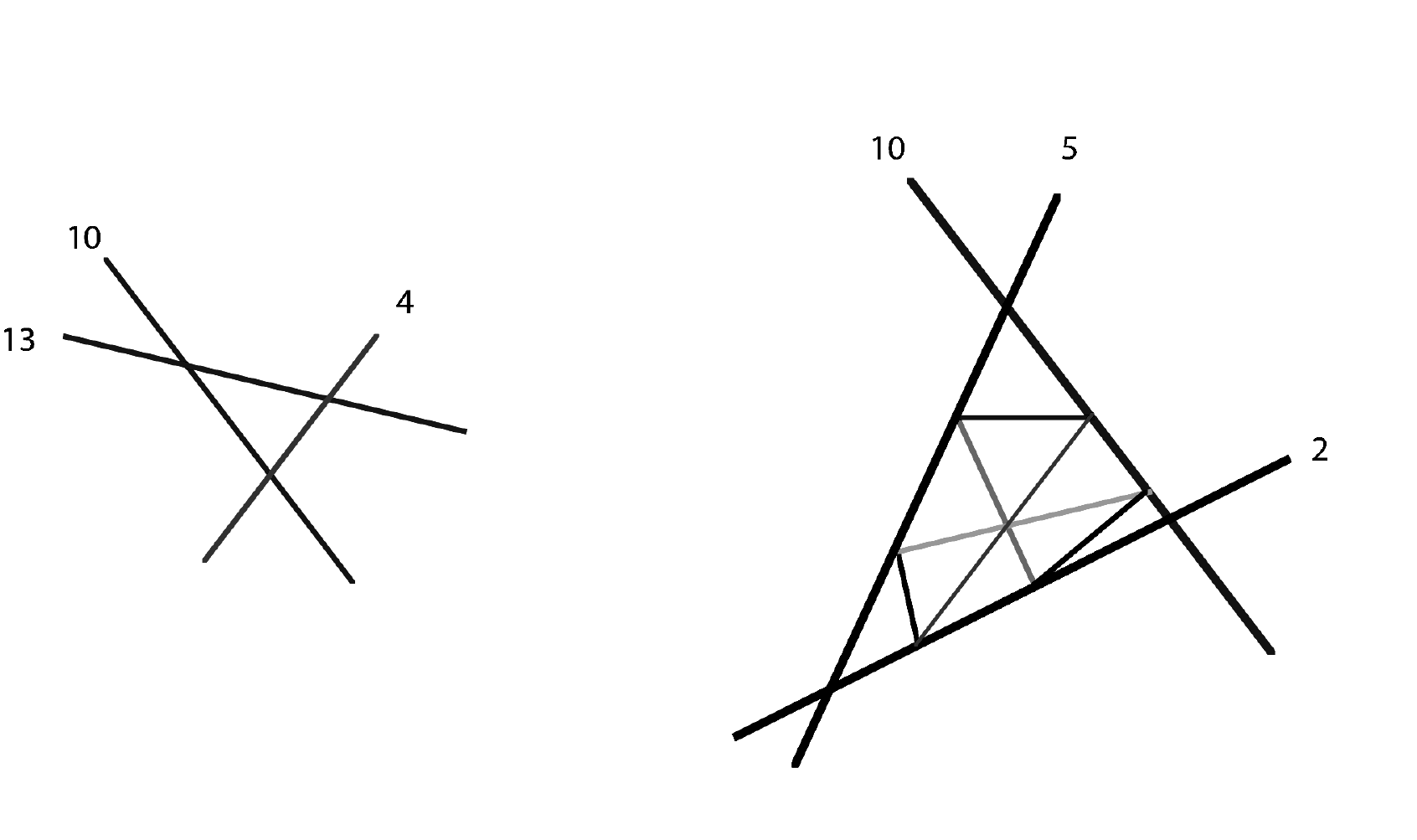}
\caption{\label{label} The segments $G^{(0)}_{\mu, 14}$, for $\mu=4, 10,13$ determine the elementary triangle $F:= \Delta_{14}^{(0)}(4, 10, 13)$. A dissection of  $\iota_{14,3} \Delta_{14}^{(0)}(4, 10, 13) \cong \Delta_{14}^{(0)}(10, 2, 5) $ is shown on the right.}
\end{figure}
    \par
  \begin{figure}[h]
 \includegraphics[width=25pc]{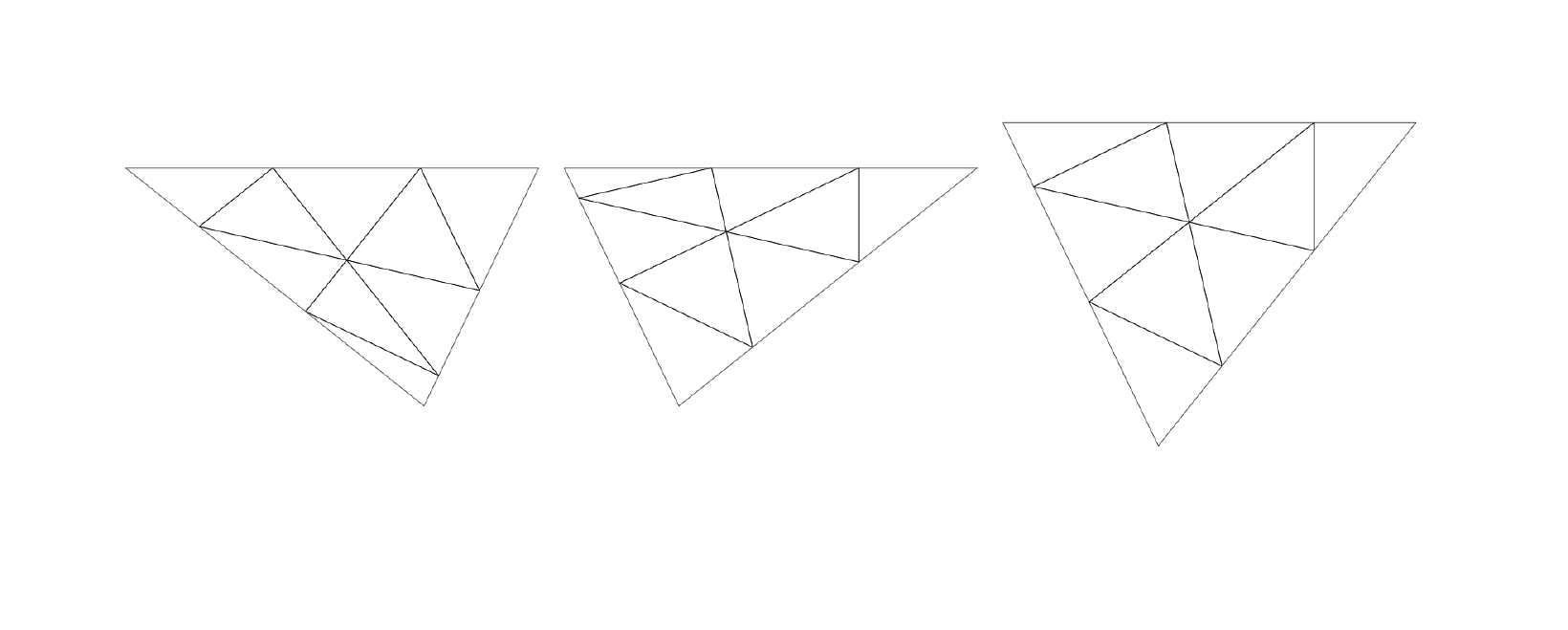}
\caption{\label{label} Substitution rules for $F, \hat{F}, G$ corresponding to $\iota_{14,3}$.}
\end{figure}
\par

    \par
  \begin{figure}[h]
 \includegraphics[width=20pc]{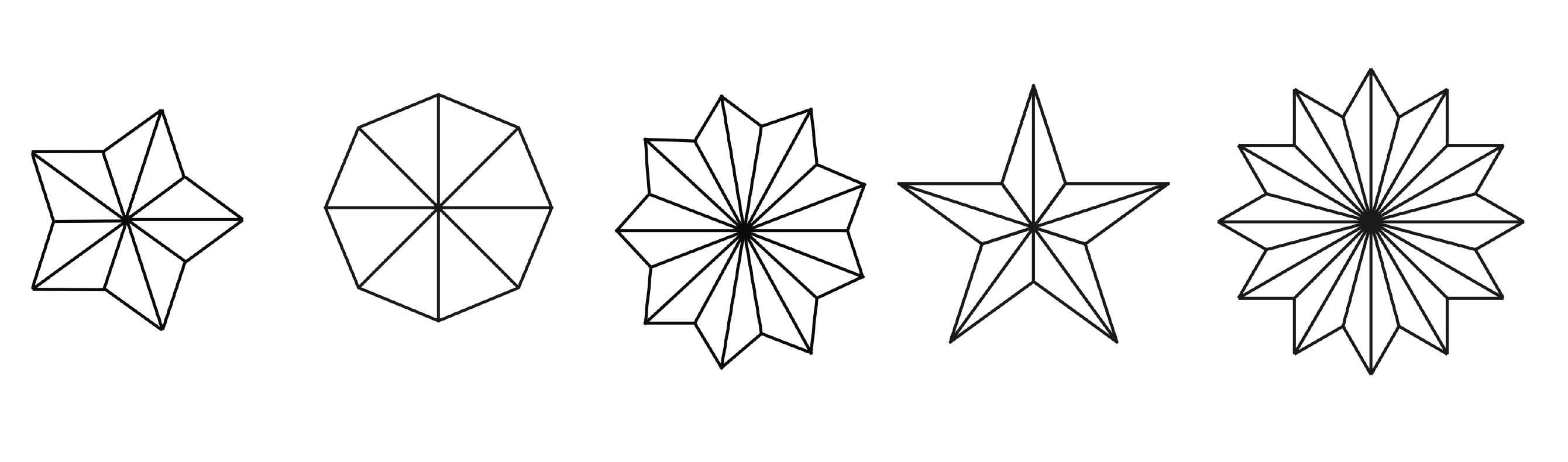}
\caption{\label{label}  Vertex configurations with rotational symmetry appearing in some examples of tilings with inflation factors $\iota_{5,2}$, $\iota_{8,3}$, $\iota_{9,4}$, $\iota_{10,3}$, $\iota_{12,2}^2$.}
\end{figure}
\par

\par
 \begin{defn}
Given $d$ and $p$, for every prototile $T \in \mathcal{F}_{d}$ we define the inflation rule $\Phi_{d,p,+}(T)$ by dissecting $\iota_{d,p}\cdot T$ according to 
Eqs. (3.6)-(3.12) and $$\Phi_{d,p,-}(T):=\Phi_{d,p,+}(\widetilde{\hat{T}}), \Phi_{d,p,-}(\widetilde{T}):=\Phi_{d,p,+}(\hat{T}),\\ \Phi_{d,p,-}(\hat{T}):=\Phi_{d,p,+}(\widetilde{T}), \Phi_{d,p,-}(\widetilde{\hat{T}}):=\Phi_{d,p,+}(T)$$ if $T$ is non-isosceles and $$\Phi_{d,p,-}(T):=\Phi_{d,p,+}(\widetilde{T}), \\ \Phi_{d,p,-}(\widetilde{T}):=\Phi_{d,p,+}(T)$$ otherwise, where $\widetilde{T}$ is related to $T$ through a reflection, including the ID in the reflection.
 \end{defn}
   \bigskip\par
{\it Example}.  In Fig.4 we have shown an example for $d=14$. The dissection of  $\iota_{14,3} \Delta_{14}^{(0)}(4, 10, 13)$ is determined by Eq. (3.6)  and is given by the interior of $\Delta_{14}^{(0)}(10, 2, 5)$ in $\mathcal{G}_{\Delta,14}^{(0)}$, which is shown in Fig.4 (right). The substitution rules $\Phi_{14,3,+}(T)$ are  

$$\Phi_{14,3,+}(A)=\hat{A}\cup\widetilde{O}\cup\hat{B} \cup \widetilde{N}, \Phi_{14,3,+}(\hat{A})=P\cup\widetilde{\hat{A}} \cup O \cup \widetilde{\hat{C}}, \Phi_{14,3,+}(B)=\widetilde{P}\cup\hat{A}\cup\widetilde{O}\cup\hat{C} \cup \widetilde{M}$$
$$\Phi_{14,3,+}(\hat{B})=\widetilde{\hat{B}}\cup O\cup \widetilde{\hat{C}} \cup\hat{N} \cup \widetilde{\hat{K}}, \Phi_{14,3,+}(C)=\widetilde{O}\cup\hat{B}\cup\widetilde{N}\cup\hat{D} \cup \widetilde{I}\cup\widetilde{M} \cup \hat{C} \cup \widetilde{\hat{N}}, $$
$$\Phi_{14,3,+}(\hat{C})=\widetilde{\hat{A}}\cup O\cup \widetilde{\hat{C}} \cup M \cup \widetilde{\hat{E}} \cup \widetilde{\hat{D}} \cup N  \cup \widetilde{\hat{B}}, \Phi_{14,3,+}(D)=\widetilde{O}\cup\hat{C}\cup\widetilde{M}\cup\hat{E} \cup \widetilde{H}\cup\widetilde{\hat{I}} \cup \hat{K} \cup \widetilde{\hat{N}}, $$
$$\Phi_{14,3,+}(\hat{D})=\widetilde{\hat{C}}\cup \hat{N}\cup \widetilde{\hat{K}} \cup \hat{I} \cup \widetilde{L} \cup \widetilde{\hat{E}} \cup M  \cup \widetilde{\hat{C}} \cup \widetilde{\hat{D}}, \Phi_{14,3,+}(E)=\widetilde{N}\cup\hat{D}\cup\widetilde{I}\cup\hat{F}\cup\widetilde{\hat{F}} \cup \widetilde{H}\cup\hat{E} \cup \widetilde{M} \cup \widetilde{\hat{I}}, $$
$$\Phi_{14,3,+}(\hat{E})=\widetilde{\hat{C}}\cup M\cup \widetilde{\hat{E}} \cup H \cup \widetilde{G} \cup \widetilde{\hat{F}} \cup I  \cup \widetilde{\hat{D}} \cup \widetilde{J}, \Phi_{14,3,+}(F)=\widetilde{M}\cup\hat{E}\cup\widetilde{H}\cup G\cup\widetilde{G} \cup \widetilde{F} \cup L  \cup \widetilde{\hat{I}} \cup \widetilde{\hat{J}}, $$
$$\Phi_{14,3,+}(\hat{F})=\widetilde{\hat{K}}\cup\hat{I} \cup\widetilde{L}\cup F \cup\widetilde{E}\cup \widetilde{\hat{E}} \cup H \cup\widetilde{G}\cup \widetilde{\hat{F}}, \Phi_{14,3,+}(G)=\widetilde{I}\cup\hat{F}\cup \widetilde{\hat{F}} \cup H \cup \widetilde{\hat{E}} \cup\widetilde{H} \cup G\cup\widetilde{G}\cup\widetilde{F}  $$
$$\Phi_{14,3,+}(H)=\widetilde{D} \cup E  \cup\widetilde{F} \cup G \cup\widetilde{H} \cup  \widetilde{G} \cup F \cup\widetilde{E} \cup \widetilde{L}, 
\Phi_{14,3,+}(I)=\widetilde{C} \cup D  \cup\widetilde{E} \cup F \cup\widetilde{G} \cup  \widetilde{L} \cup \hat{J} \cup\widetilde{K} ,$$
$$ \Phi_{14,3,+}(\hat{I})=\widetilde{\hat{I}}\cup L \cup \widetilde{\hat{J}} \cup K \cup \widetilde{F} \cup E \cup\widetilde{D} \cup\widetilde{E},  
\Phi_{14,3,+}(J)=\widetilde{\hat{E}} \cup \hat{I} \cup \widetilde{L} \cup \hat{J}  \cup  \widetilde{K}, \Phi_{14,3,+}(\hat{J})=\widetilde{\hat{N}} \cup \hat{K} \cup \widetilde{\hat{I}} \cup L  \cup  \widetilde{F}, $$
$$\Phi_{14,3,+}(K)=\widetilde{M} \cup  \hat{D}  \cup\widetilde{I} \cup J \cup\widetilde{J}, 
\Phi_{14,3,+}(\hat{K})=\widetilde{\hat{B}} \cup N  \cup \widetilde{\hat{D}} \cup I \cup \widetilde{\hat{F}}, \Phi_{14,3,+}(L)=\widetilde{\hat{D}} \cup  I  \cup\widetilde{J} \cup J \cup \widetilde{\hat{F}} \cup \hat{F} \cup\widetilde{I} \cup\widetilde{H},$$
$$\Phi_{14,3,+}(M)=\widetilde{B} \cup C  \cup \widetilde{D} \cup E \cup \widetilde{F} \cup  \widetilde{E} \cup D \cup \widetilde{C} \cup \widetilde{K}, 
\Phi_{14,3,+}(N)=\widetilde{A} \cup B  \cup \widetilde{C} \cup D \cup \widetilde{E},$$
$$\Phi_{14,3,+}(\hat{N})=\widetilde{C} \cup C  \cup \widetilde{D} \cup K \cup \widetilde{\hat{J}}, \Phi_{14,3,+}(O)= A \cup  \widetilde{B} \cup C \cup \widetilde{D}  \cup \widetilde{A}  \cup B \cup  \widetilde{C}, 
\Phi_{14,3,+}(P)= A  \cup \widetilde{B} $$

  \bigskip \par
We have only given the prototile content in the substitutions rules. A more precise description can be given in terms of formal languages as in \cite{esc04, esc11} and references within. The substitutions $\Phi_{14,3,+}(T)$ for $T= F, \hat{F}, G$ are represented in Fig.5. 
 \begin{defn}

By $\bold{S}( \mathcal{F}_{d};\Phi_{d,p,\epsilon})$ ($\epsilon=\pm$)  we denote the set of all tilings $\mathcal{P}$ of the entire plane, where every patch of $\mathcal{P}$ is congruent to some patch in some $\Phi_{d,p,\epsilon}^{n}(T)$, with $T \in \mathcal{F}_{d}$.

 \end{defn}
 As a consequence of the previous results we have
 \begin{thm}
Assume $d \in {\Bbb{N}}, d>4, p=2,3...\lfloor{\frac{d}{2}}\rfloor, \epsilon=\pm$. The set $\bold{S}( \mathcal{F}_{d};\Phi_{d,p,\epsilon})$ is non empty. It consists of  tilings of the plane by tiles being congruent to members of  $\mathcal{F}_{d}$. 
\end{thm}
   \bigskip\par

 A vertex configuration is a patch where all the tiles have a vertex in common. In previous papers we have studied several particular cases of the construction presented here (see a more exhaustive list of references in \cite{esc17}). In Fig.6 we show vertex configurations with rotational symmetry appearing in the  studied examples of tilings with inflation factors $\iota_{5,2}$, $\iota_{8,3}$, $\iota_{9,4}$, $\iota_{10,3}$, $\iota_{12,2}^2$, which are PV numbers.
    \par
    {\it Example}. For $d=14$ we represent in Figs.7, 8  some patches of tilings with inflation factors $\iota_{14,3}$ and $\iota_{14,5}$. A fragment of $\Phi_{14,3,+}^{3}(G)$ is shown in Fig.7 (left) and one of $\Phi_{14,3,-}^{5}(G)$ in the same figure (right). In Fig.8 we can see fragments corresponding to $\Phi_{14,5,+}^{3}(G)$ and $\Phi_{14,5,-}^{4}(G)$ ($\iota_{14,5}$ is also a PV number \cite{esc08}).

    \par
  \begin{figure}[h]
 \includegraphics[width=12pc]{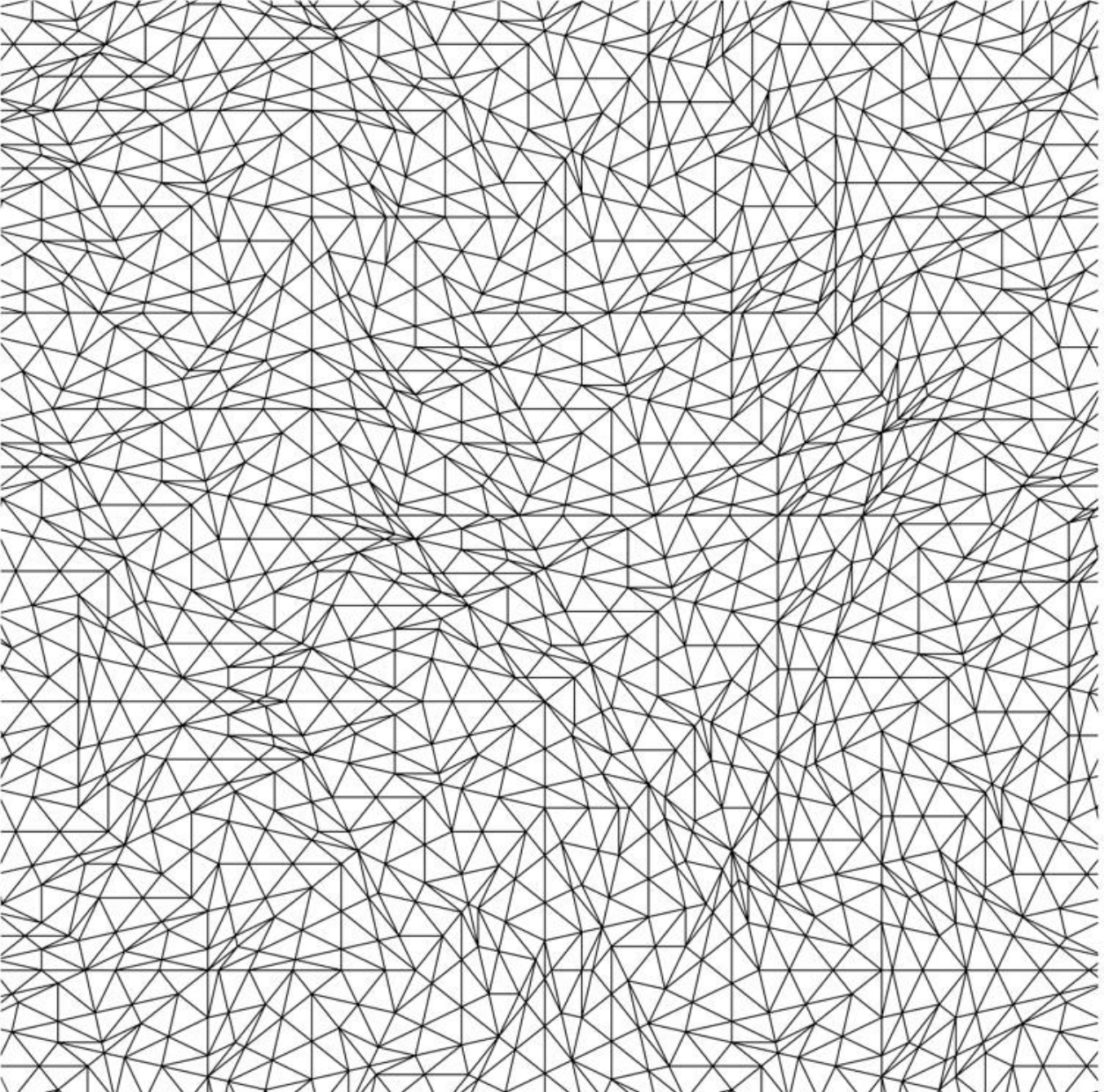}
  \includegraphics[width=12pc]{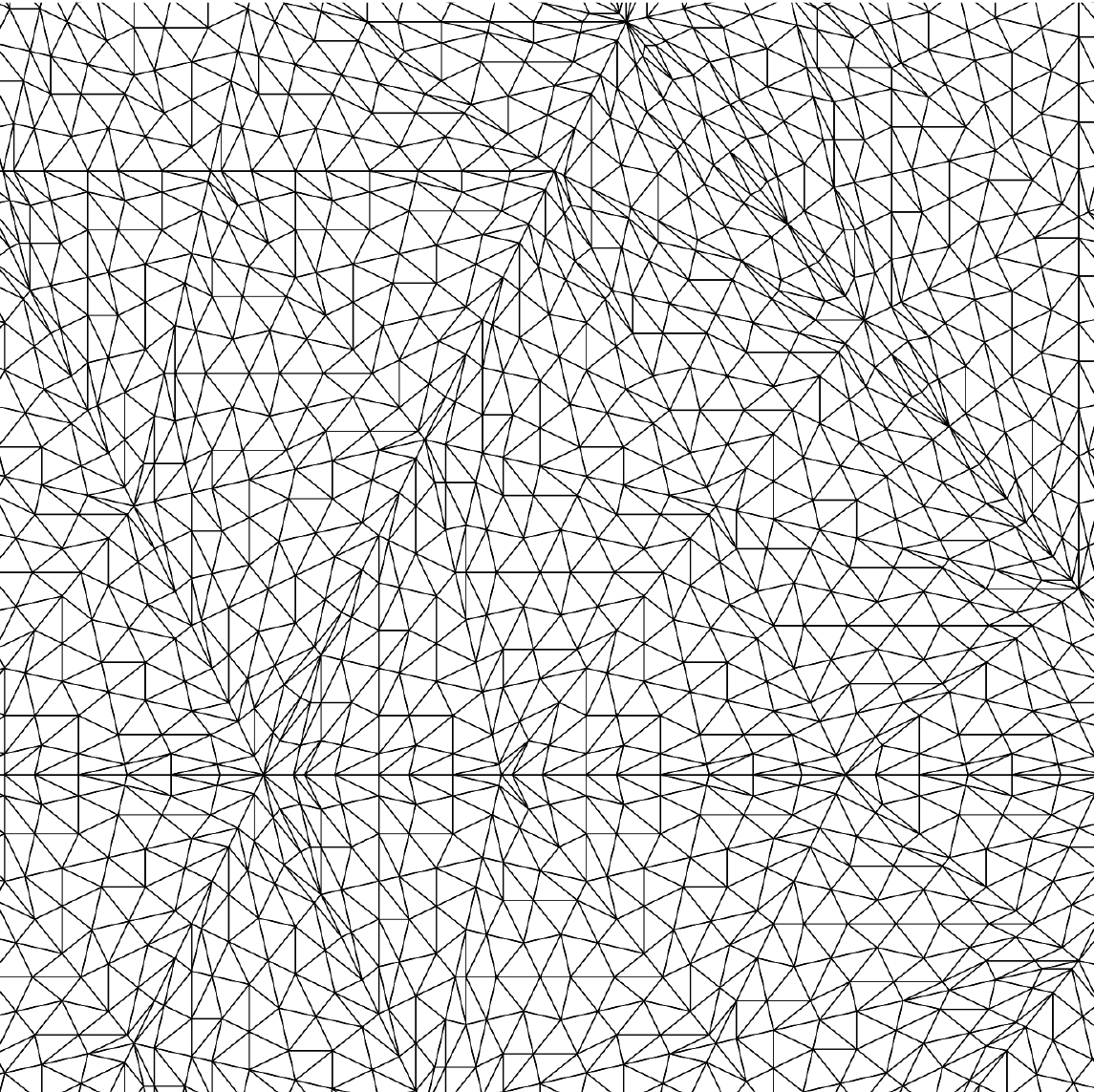}
\caption{\label{label}    A fragment of $\Phi_{14,3,+}^{3}(G)$ (left) and $\Phi_{14,3,-}^{5}(G)$ (right).}
\end{figure}
\par
    \par
  \begin{figure}[h]
 \includegraphics[width=12pc]{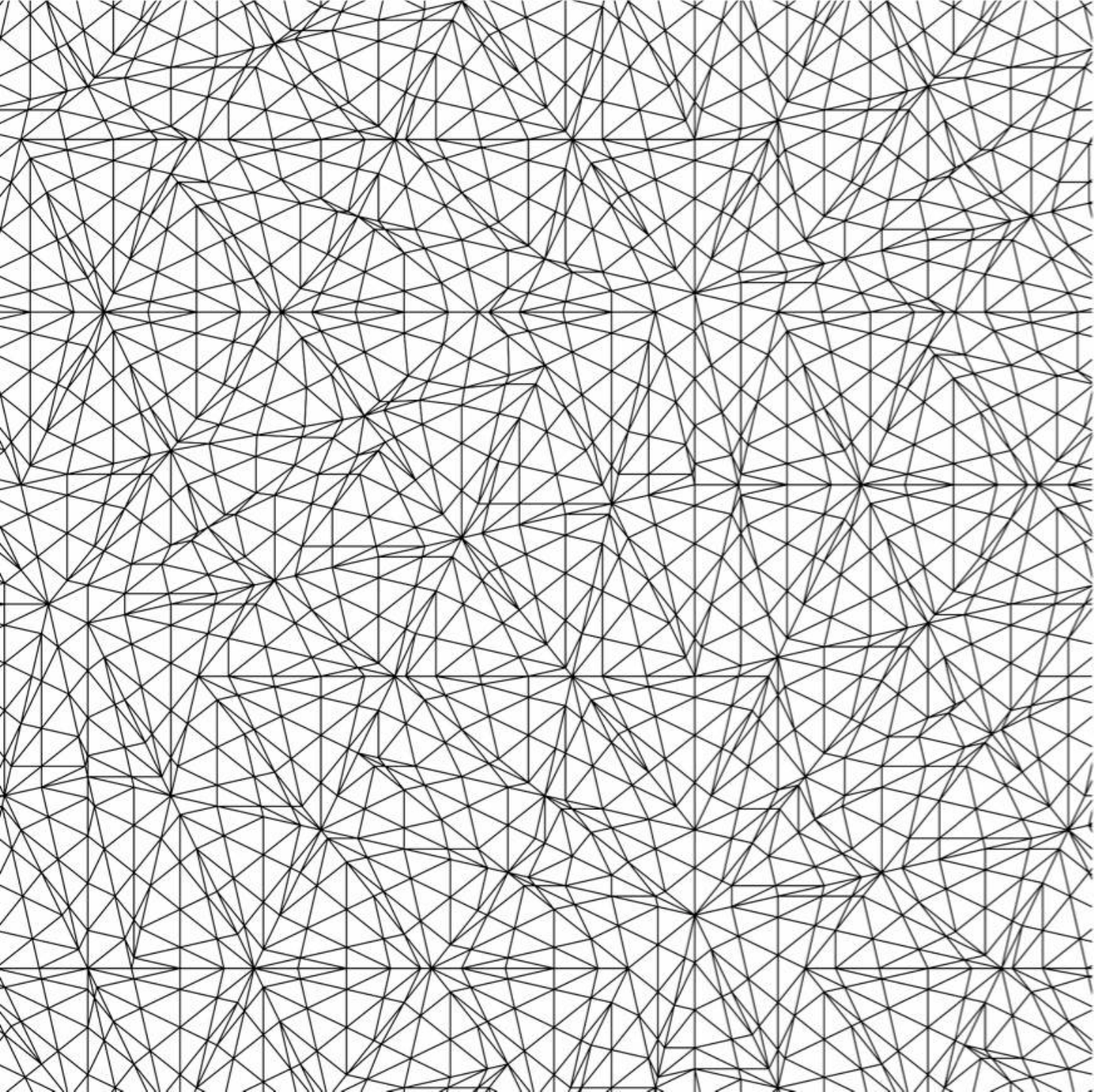}
  \includegraphics[width=12pc]{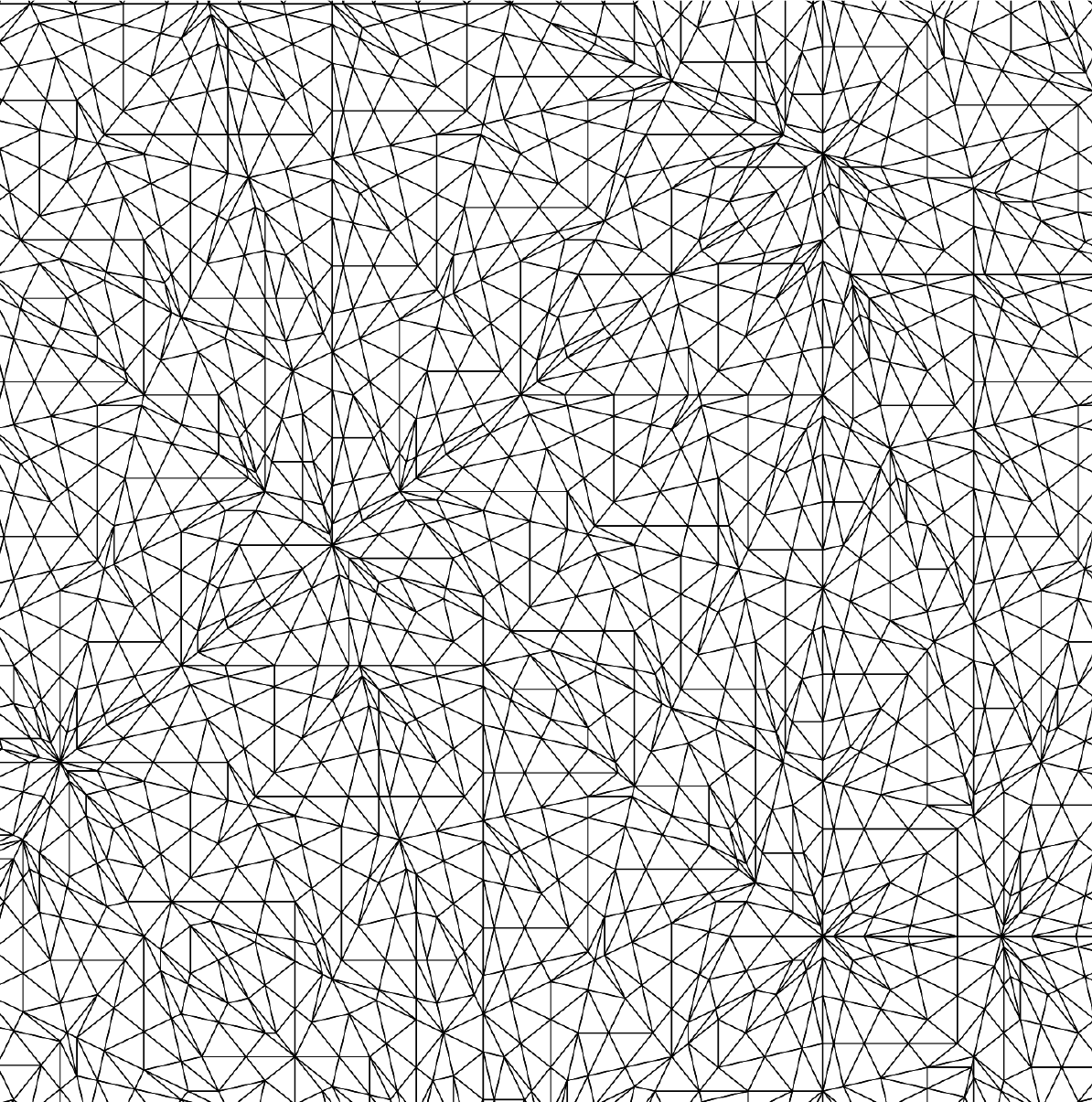}
\caption{\label{label}   Fragments of $\Phi_{14,5,+}^{3}(G)$ (left) and $\Phi_{14,5,-}^{4}(G)$ (right). }
\end{figure}
\par
    \par
  \begin{figure}[h]
 \includegraphics[width=12pc]{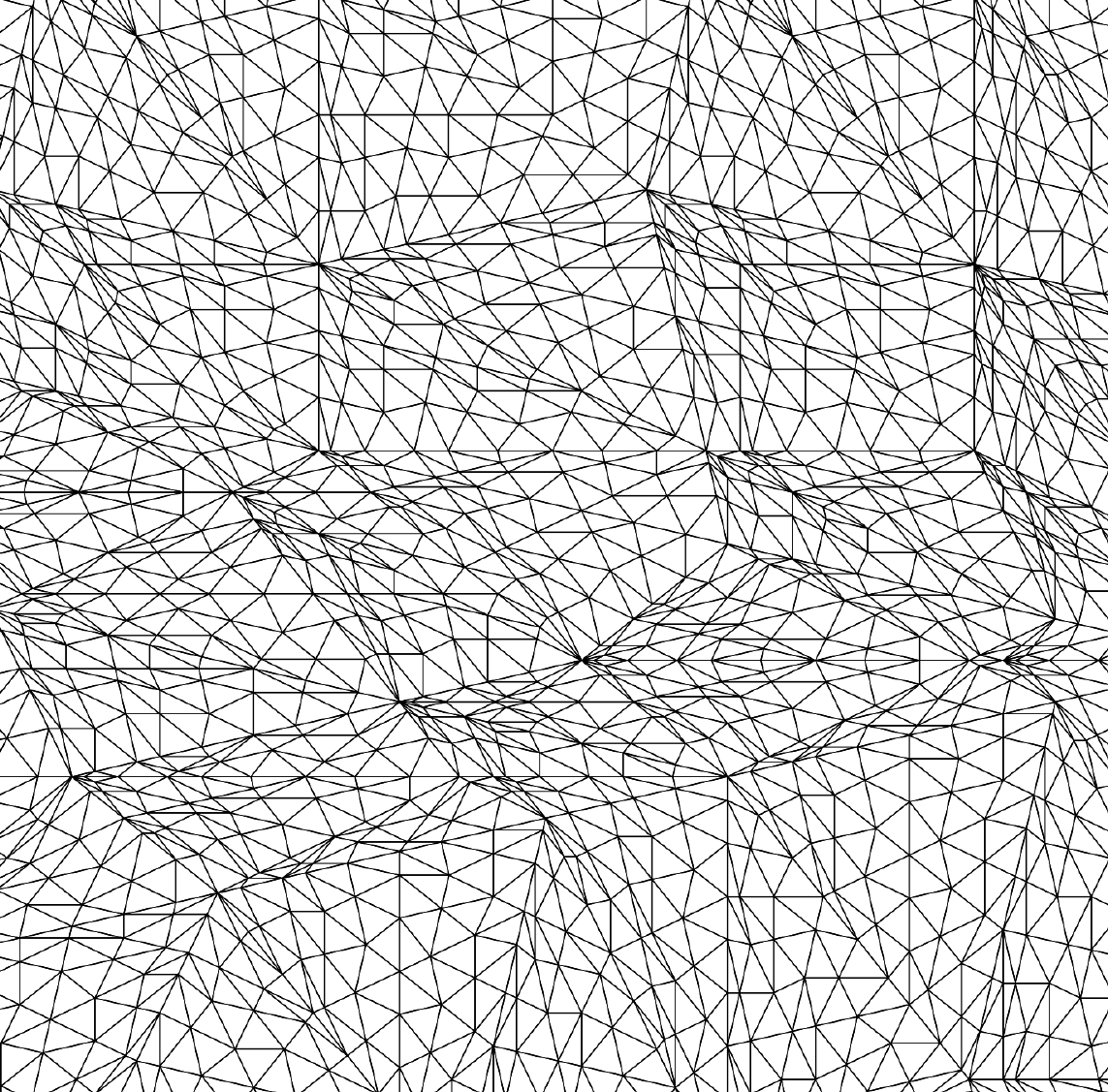}
  \includegraphics[width=12pc]{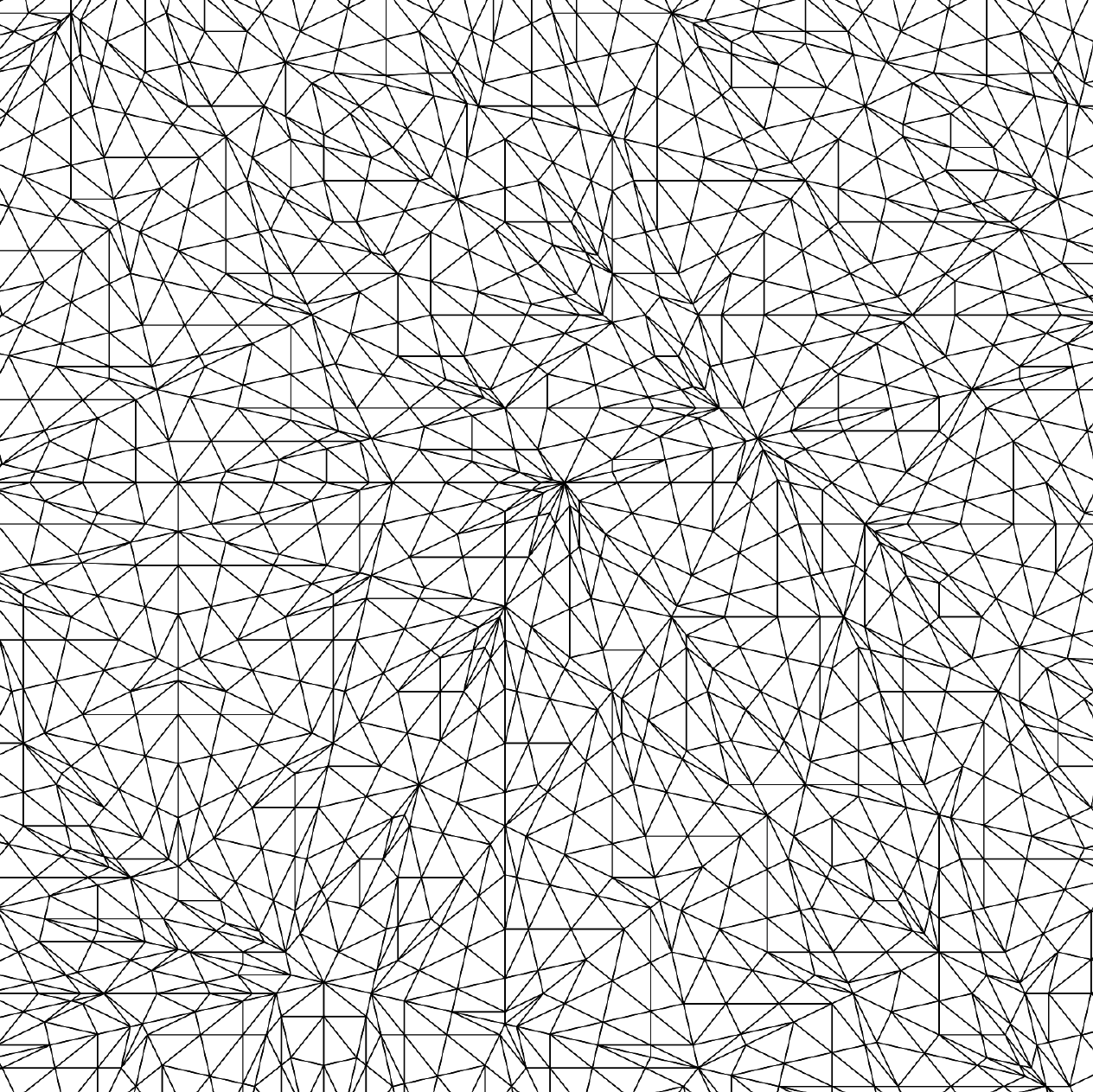}
\caption{\label{label} Parts of $(\Phi_{14,3,+}\Phi_{14,3,-})^{3}(G)$ (left) and $(\Phi_{14,5,+}\Phi_{14,5,-})^{2}(G)$ (right). }
\end{figure}
\par
    \par
  \begin{figure}[h]
 \includegraphics[width=24pc]{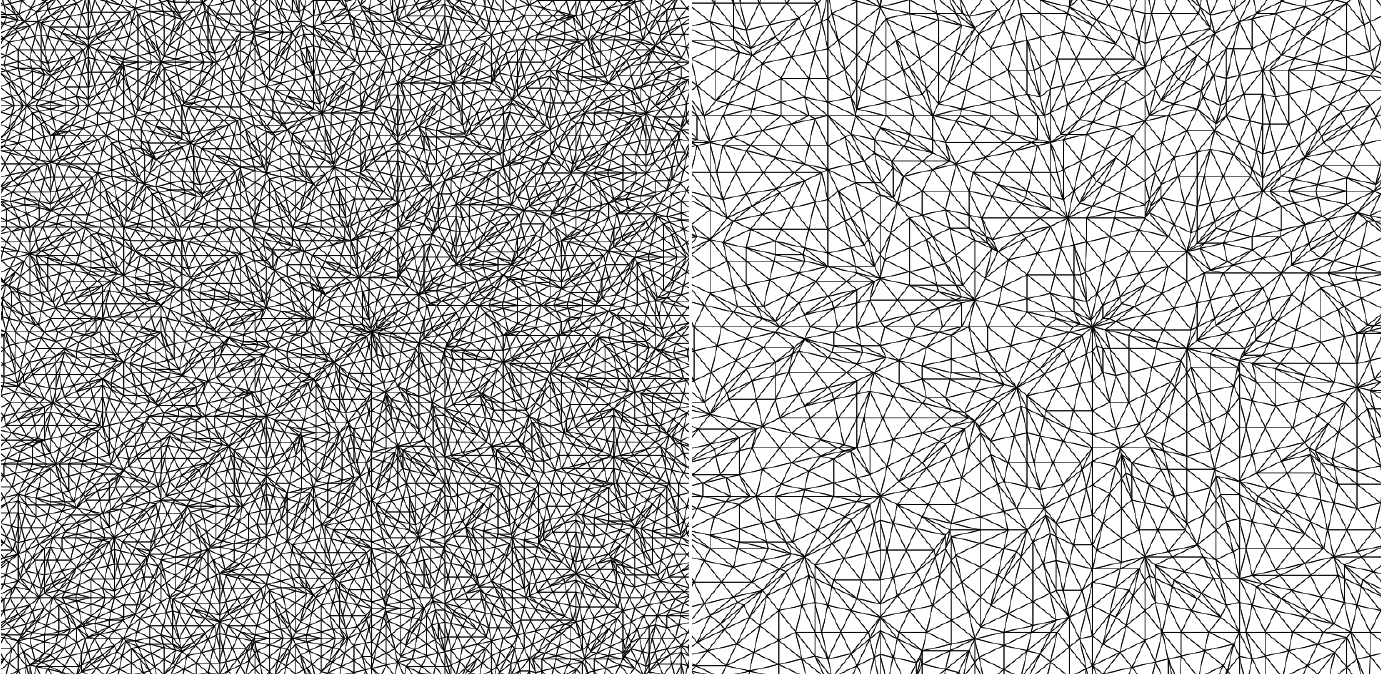}
\caption{\label{label} $(\Phi_{14,3,+}\Phi_{14,5,+})^{2}(G)$ in two different scales.}
\end{figure}
\par
    \par
  \begin{figure}[h]
 \includegraphics[width=24pc]{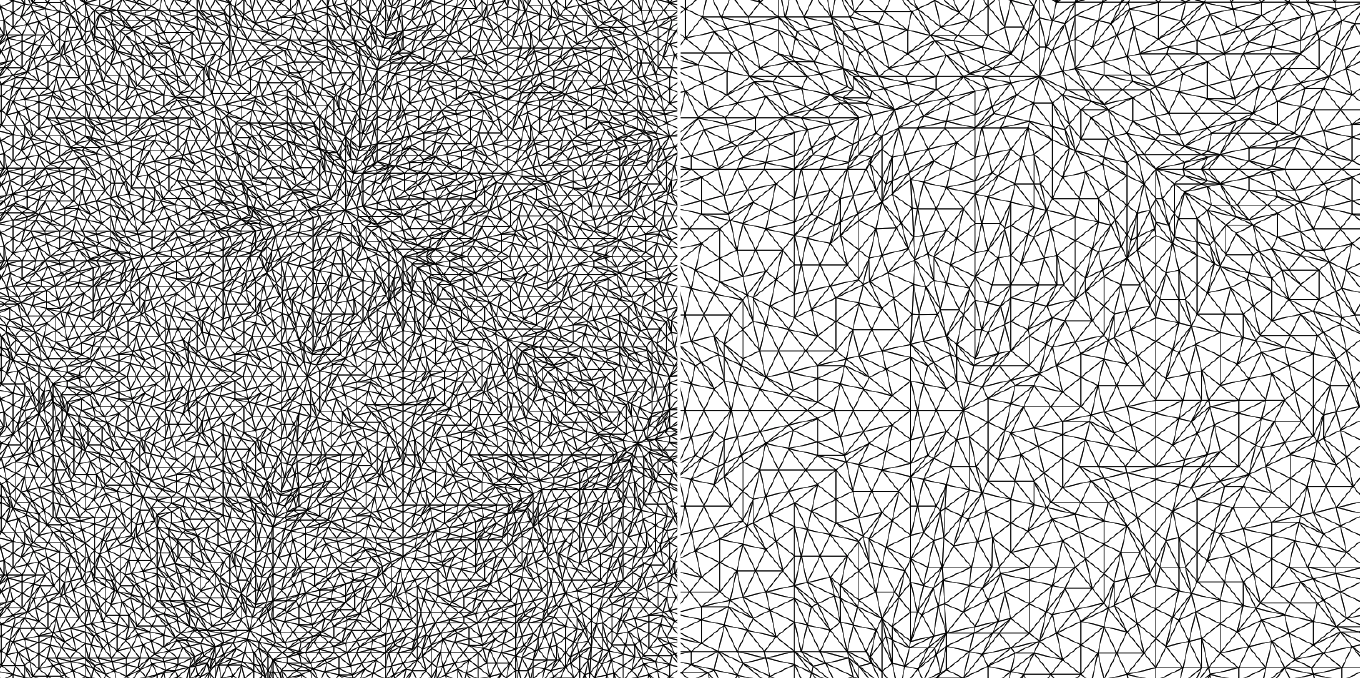}
\caption{\label{label} $(\Phi_{14,5,-}\Phi_{14,3,+})^{2}(G)$ in two scales.}
\end{figure}
\par
\par
 \section{Random tilings for $d=2q$}
 Multiple different substitution rules defined on the same set of prototiles have been studied in \cite{god89, nis96, esc08, gah15}. The method of composition of inflation rules, also called multisubstitutions, consists in applying the same inflation rule to each tile in a given inflation step. 
 In Fig.9 we can see parts of $(\Phi_{14,3,+}\Phi_{14,3,-})^{3}(G)$ and $(\Phi_{14,5,+}\Phi_{14,5,-})^{2}(G)$. Also composition of substitution rules corresponding to different inflation factors are possible as shown in Fig.10  for $(\Phi_{14,3,+}\Phi_{14,5,+})^{2}(G)$ and Fig. 11 for $(\Phi_{14,5,-}\Phi_{14,3,+})^{2}(G)$, where two different magnifications are used for the same fragment ($\iota_{14,3}\iota_{14,5}$ is a PV number \cite{esc08}). In this section we treat the problem of the generation of random tilings by means of prototile rearrangements.
 \par
 \begin{defn}
If a quadrilateral $\square$ with diagonals $S_{i}, S_{j}$ is formed by the union of two elementary triangles sharing $S_{i}$, the tile rearrangement consists in replacing $S_{i} \longrightarrow S_{j}$ and $\square$ is now the union of two different prototiles.
 \end{defn} 
  \par
  The set of prototiles without ID is denoted by  $\bar{\mathcal{F}}_{d}$.
  \par
  \begin{defn}
The random tiling ensemble denoted by $\bold{R_{r}}(\bar{\mathcal{F}}_{d};\Phi_{d,p,\epsilon})$ is the set of tilings obtained by tile rearrangements in any member of $\bold{S}( \mathcal{F}_{d};\Phi_{d,p,\epsilon})$.
  \end{defn} 
  \par
 We construct, in the interior of $\mathcal{G}_{\Delta,2q}^{(\kappa)}$, the regular polygon with vertices $\{p_{c,q+c}\}_{c=0,1, 2,..., q-1}$ \cite{esc11}. Its edges have length $4s_{1}s_{q-1}$ and some of them may lie on a $G^{(\kappa)}_{\mu,2q}$ (see Fig.1 for $q=7$). The edge with vertices $p_{c,q+c}, p_{c+1,q+c+1}$ lie in some $G^{(\kappa)}_{\mu,2q}$ for $\kappa=0$ when $3 c \equiv q-1$ (mod $2q$) or $3 c \equiv q-2$ (mod $2q$), namely $(c, \mu) \in \{ (l,l+1), (2l,5l+1)\}$ if $q=3l+1$  and  $(c, \mu) \in \{ (l,l), (2l+1,5l+4)\}$ if $q=3l+2$. There is no solution for $c$ when $q=3l$ if $\kappa=0$. For  $\kappa =-2$, $c$ must be a solution of $3 c \equiv q-2$ (mod $2q$) or $3 c \equiv q-3$ (mod $2q$), therefore $(c, \mu)  \in \{ (l-1,l), (3l-1,3l), (5l-1,5l)\}$ if $q=3l$. 
 
 The sides not lying on a $G^{(\kappa)}_{\mu,2q}$ define tile rearrangements in the substitution rules \cite{esc08} associated to edge flips $S_{q} \longrightarrow S_{q-1}$. The set of substitution rules $\Phi_{d, R} := \{\Phi_{d, 1}, \Phi_{d, 2}, ...\Phi_{d, k}\}$ is obtained by considering all possible combinations of tile rearrangements in the inflated prototiles.
 \par
 Random substitution tilings are characterised by the fact that one can apply at each inflation step different substitution rules to each tile \cite{esc04, gah15}. 
 \par
 \begin{defn}
 Given a set of substitution rules $\Phi_{d, R} := \{\Phi_{d, 1}, \Phi_{d, 2}, ...\Phi_{d, k}\}$, we define the random substitution tiling ensemble corresponding to a prototile set $ \mathcal{F}$, denoted by $\bold{R_{s}}( \mathcal{F} ;\Phi_{d, R})$, as the set of tilings constructed in such a way that one can apply to each tile and at each inflation step different substitution rules belonging to $\Phi_{d, R}$.
 \end{defn} 
\par
  \begin{thm}
Assume $d =2q \in {\Bbb{N}}, q>2 $. The sets $\bold{R_{r}}( \bar{\mathcal{F}}_{d};\Phi_{d,p,\epsilon})$ and $\bold{R_{s}}( \bar{\mathcal{F}}_{d};\Phi_{d, R})$, with prototiles congruent to members of $ \bar{\mathcal{F}}_{d}$, are non-empty. The inflation factor associated with $\bold{R_{s}}( \bar{\mathcal{F}}_{d};\Phi_{d, R})$ is $\iota_{d,q}$. 
\end{thm}
    \par
\begin{proof}
 \par
1) $\kappa=0, q=3l+1$: each pair of consecutive polygon vertices $p_{l+a,4l+a+1}, p_{l+a+1,4l+a+2} $ define a quadrilateral $\square$ in $\mathcal{G}_{\Delta,2q}^{(0)}$ which is the union of two adjacent prototiles:
 $$\square^{(0)}_{6l+2}(l+a,4l+a+1;l+a+1,4l+a+2)= \Delta^{(0)}_{6l+2}(l+a, 4l+a+1,l-2a)\cup \Delta^{(0)}_{6l+2}(l+a+1, 4l+a+2,l-2a)$$
We introduce edge flips consisting in replacing the common edge of the prototiles by the polygon edge when $a \in \{1,2,..., l-1, l+1,..., 3l\}$ which produces a prototile rearrangement. We denote by $\Delta' $ a triangle not belonging to $\mathcal{G}_{\Delta,2q}^{(\kappa)}$ but congruent to some prototile and by $\mu'$ the index of a segment parallel to $G^{(\kappa)}_{\mu,2q}$ not belonging to $\mathcal{G}_{\Delta,2q}^{(\kappa)}$. After the edge flip we have 
$$\square^{(0)}_{6l+2}(l+a,4l+a+1;l+a+1,4l+a+2)= \Delta'(l+a, 4l+a+2,(l+4a+1) ')\cup \Delta'(4l+a+1, l+a+1,(l+4a+1) ')$$ with $\Delta'(l+a, 4l+a+2,(l+4a+1) ') \cong \Delta^{(0)}_{6l+2}(3l-a, 3l+2a+1,6l-a+2)$ and  $\Delta'(4l+a+1, l+a+1,(l+4a+1) ') \cong \Delta^{(0)}_{6l+2}(3l-a+1, 3l+2a+1,6l-a+1)$. 
 \bigskip\par
2) $\kappa=0, q=3l+2$: the vertices $p_{l+a,4l+a+2}, p_{l+a+1,4l+a+3} $ define a quadrilateral :
 $$\square^{(0)}_{6l+4}(l+a,4l+a+2;l+a+1,4l+a+3)= \Delta^{(0)}_{6l+4}(l+a, 4l+a+2,l-2a+1)\cup \Delta^{(0)}_{6l+4}(l+a+1, 4l+a+3,l-2a+1)$$
 and the edge flip for $a \in \{1,2,..., l, l+2,..., 3l+1\}$ gives
 $$\square^{(0)}_{6l+4}(l+a,4l+a+2;l+a+1,4l+a+3)= \Delta'(l+a, 4l+a+3,(l+4a) ')\cup \Delta'(l+a+1,4l+a+2, (l+4a) ')$$ with $\Delta'(l+a, 4l+a+3,(l+4a) ')\cong \Delta^{(0)}_{6l+4}(3l-a+2, -a+1,3l+2a+2)$ and  $ \Delta'(l+a+1, 4l+a+2, (l+4a) ') \cong \Delta^{(0)}_{6l+4}(3l-a+3, -a,3l+2a+2)$
 \bigskip\par
3) $q=3l$.
 \bigskip\par
3.1) $\kappa=0$: the vertices $p_{l+a,4l+a}, p_{l+a+1,4l+a+1} $ define a quadrilateral 
 $$\square^{(0)}_{6l}(l+a,4l+a;l+a+1,4l+a+1)= \Delta^{(0)}_{6l}(l+a, 4l+a,l-2a-1)\cup \Delta^{(0)}_{6l}(l+a+1, 4l+a+1,l-2a-1)$$
An edge flip for $a \in \{1,2,..., 3l-1\}$ gives 
 $$\square^{(0)}_{6l}(l+a,4l+a;l+a+1,4l+a+1)= \Delta'(4l+a, l+a+1,(l+4a+2) ')\cup \Delta'(4l+a,l+a+1, (l+4a+2) ')$$
 In this case we do not have congruent prototiles $\Delta^{(0)}_{6l}(\lambda, \mu,\nu)$  inside $\mathcal{G}_{\Delta,6l}^{(0)}$ with, for instance, $\sigma_{6l}=1$, because $3 \mu \equiv -3a+1$ (mod $6l$) has no solution. They are congruent to elementary triangles with $\sigma_{6l}=-3$ which belong to $\mathcal{G}_{\Delta,6l}^{(-2)}$: $\Delta'(4l+a, l+a+1,(l+4a+2) ')\cong \Delta^{(-2)}_{6l}(3l-a-2, 6l-a-1,3l+2a)$ and $\Delta'(4l+a,l+a+1, (l+4a+2) ')\cong \Delta^{(-2)}_{6l}(2l-a-2, 5l-a-1,5l+2a)$. 
 \bigskip\par
3.2) $\kappa=\pm 2$: in $\mathcal{G}_{\Delta,6l}^{(-2)}$ the vertices $p_{l+a-1,4l+a-1}, p_{l+a,4l+a} $ define a quadrilateral 
 $$\square^{(-2)}_{6l}(l+a-1,4l+a-1;l+a,4l+a)= \Delta^{(-2)}_{6l}(l+a-1, 4l+a-1,l-2a-1)\cup \Delta^{(-2)}_{6l}(l+a, 4l+a,l-2a-1)$$
An edge flip for $a \in \{1,2,..., l-1, l+1,..., 2l-1, 2l+1, ..., 3l-1\}$ gives 
 $$\square^{(-2)}_{6l}(l+a-1,4l+a-1;l+a,4l+a)= \Delta'(l+a-1, 4l+a,(l+4a) ')\cup \Delta'(l+a, 4l+a-1, (l+4a) ')$$
 with $\Delta'(l+a-1, 4l+a,(l+4a) ')\cong \Delta^{(-2)}_{6l}(3l-a-1, 6l-a,3l+2a)$ and $\Delta'(l+a, 4l+a-1, (l+4a) ')\cong \Delta^{(-2)}_{6l}(5l-a, 2l-a-1,5l+2a)$. If we change the signs we get a similar result for the edge flips in $\mathcal{G}_{\Delta,6l}^{(2)}$.
 \par
In  $\bold{R_{r}}( \bar{\mathcal{F}}_{d};\Phi_{d,p,\epsilon})$, where the tilings are obtained by tile rearrangements in a deterministic tiling, we have random tilings associated with all inflation factors. For $\bold{R_{s}}( \bar{\mathcal{F}}_{d};\Phi_{d, R})$, with $d=2q$, if the inflation factor is $\iota_{d,q}$, then in the sequence $$(S_{q-j+1},S_{q-j+3},...,S_{q+j-1}), 1\le j\le q$$ we have $$S_{q+j-1}=S_{d-(q+j-1)}=S_{q-j+1}, S_{q+j-3}=S_{d-(q+j-3)}=S_{q-j+3},...$$ and therefore the sequence is palindromic. Now the ID is not needed to get face to face substitution tilings \cite {esc08}. In order to produce face to face tilings with other inflation factors, combinations of several substitutions must be considered  (see \cite{esc17}, Figs.6,7 for $d=5,9$). 
\end{proof}

   \bigskip\par
  
{\it Example}. The tile rearrangements for $d=14$ are obtained having in mind that $E\cup \widetilde{E} \cong M \cup H$, $I\cup \widetilde{J} \cong L \cup K$, $ \widetilde{\hat{E}} \cup \hat{I} \cong F \cup D$ (Fig.12 (left)) and their mirror reflections, where we consider $M, H, L, K, F, D$ without the ID, namely, as members of $\bar{\mathcal{F}}_{d}$. In order to illustrate how to get $\bold{R_{r}}( \bar{\mathcal{F}}_{14};\Phi_{14,5,+})$  we have shown in Fig.12 (right) the prototiles $E, \widetilde{E} , I, \widetilde{J} , \widetilde{\hat{E}}, \hat{I}$ in $\Phi_{14,5,+}^{2}(G)$ with different grey levels. The edge flips corresponding to the tile rearrangements are represented by a thicker segment, as in Fig.1.
   \bigskip\par
   It is possible to describe the random tilings in terms of formal grammars along the lines of  \cite{esc11}. In that approach the elements of the alphabet are letters representing  the prototiles and the set of production rules $H$ in the language would correspond in this case to the set of substitution rules $\Phi_{d, R}$, having in mind the orientation of the prototiles. For non-deterministic structures one has to introduce a function $\pi: H \mapsto (0,1]$, called the probability distribution, which maps the set of production rules into the set of production probabilities. A question that deserves further study is the role played by the choice of $\pi$ in the analysis of diffraction patterns, phase transitions and other properties of the associated structures  (see also \cite{deg19} for a recent application of weighted context-free grammars in a different domain).
\par
        \begin{figure}[h]
 \includegraphics[width=16pc]{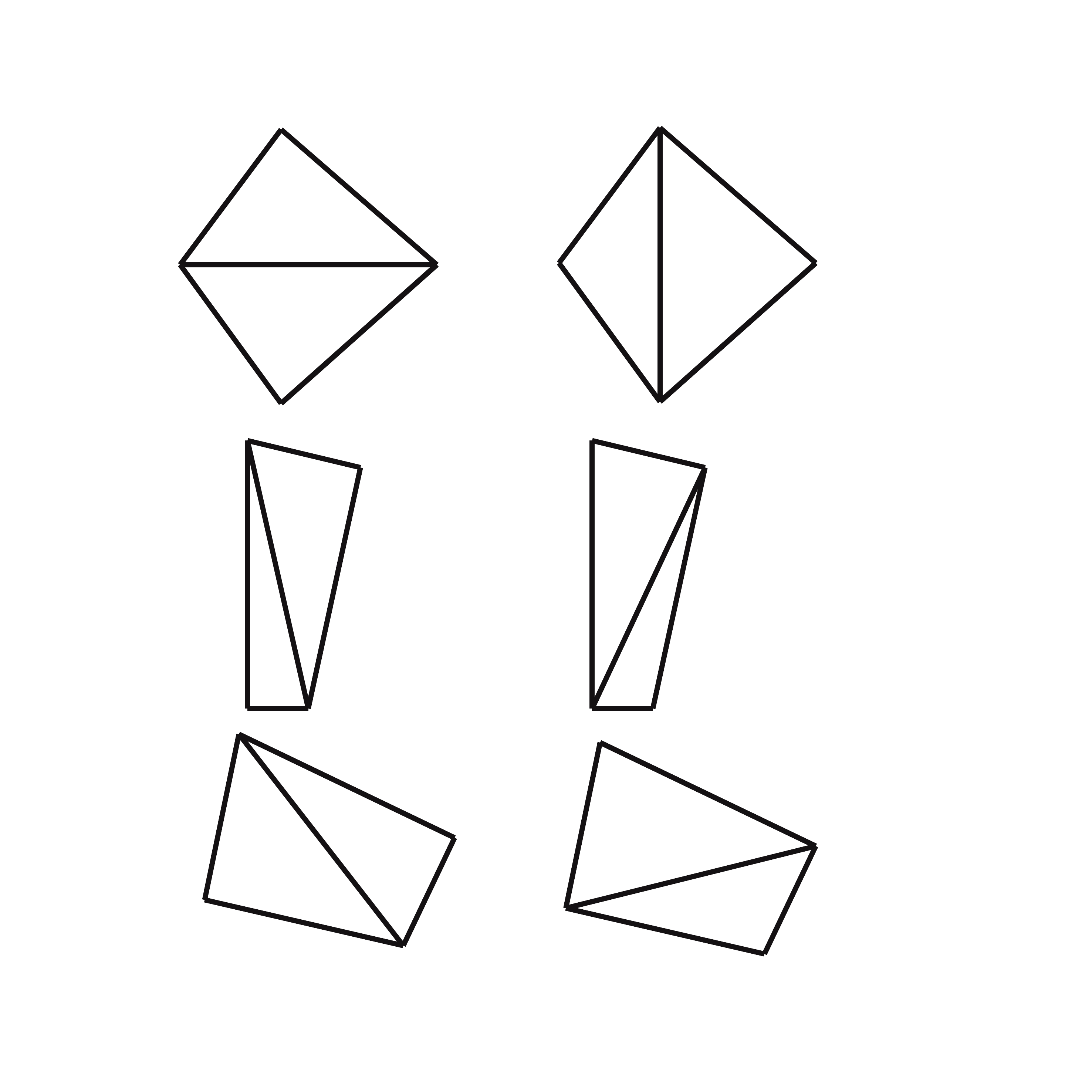}
  \includegraphics[width=20pc]{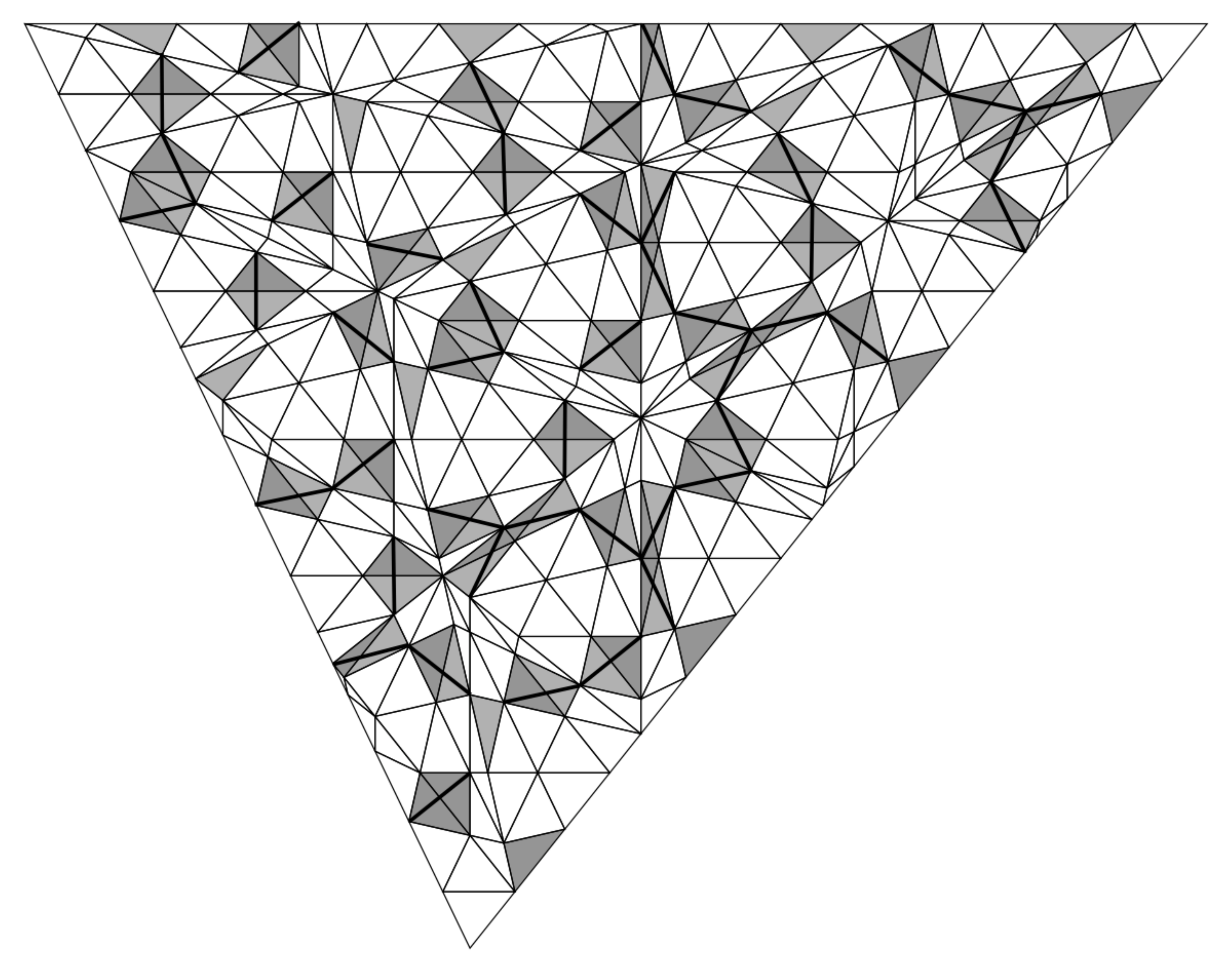}
\caption{\label{label} Tile rearrangements $E\cup \widetilde{E}\longrightarrow M \cup H$, $I\cup \widetilde{J} \longrightarrow L \cup K$, $ \widetilde{\hat{E}} \cup \hat{I} \longrightarrow F \cup D$ for $d=14$ (left). In $\Phi_{14,5,+}^{2}(G)$ the edge flips are marked with a thicker segment (right).}
\end{figure}
\par

\bigskip\par


\newpage
 \begin{thebibliography}{99}
 \bigskip\par 
\bibitem{ber66} Berger, R.: The undecidability of the Domino problem. Mem. Am. Math. Soc.. {\bf 66} (1966).
     \bigskip\par 
\bibitem{bom87} Bombieri, E. and Taylor, J.: Quasicrystals, tilings and algebraic number theory: some preliminary connections. Contemp. Math. 64, 241-264 (1987).
    \bigskip\par 
\bibitem{com74} Comtet, L.: Advanced Combinatorics, D. Reidel Publishing Company, Dordrecht, Holland (1974).
\bigskip\par 
\bibitem{deg19} DeGiuli, E.: Random language model. Phys.Rev.Lett. {\bf 122}, 128301 (2019).
\bigskip\par 
\bibitem{esc04} Escudero, J.G.: Configurational entropy for stone-inflation hexagonal and octagonal patterns. Int. J. Mod. Phys. B. {\bf 18}, 1595-1602 (2004).
\bigskip\par
\bibitem{esc08} Escudero, J.G.: Random tilings of spherical 3-manifolds. J. Geom. Phys. {\bf 58}, 1451-1464 (2008).
  \bigskip\par 
\bibitem{esc11} Escudero, J.G.: Randomness and topological invariants in pentagonal tiling spaces. Discrete Dyn. Nat. Soc. Article ID 946913 (2011). 
 \bigskip\par 
\bibitem{esc17} Escudero, J.G.: The root lattice $A_2$ in the construction of substitution tilings and singular hypersurfaces. {\it  Applications of Computer Algebra}. 
Springer Proceedings in Mathematics \& Statistics {\bf 198}, 101-117 (2017).
\bigskip\par 
\bibitem{fre98} Frettl\"{o}h, D.: Inflation\"{a}re Pflasterungen der Ebene mit $D_{2m+1}$-Symmetrie und minimaler Musterfamilie. Diploma thesis. Universit\"{a}t Dortmund (1998).
 \bigskip\par 
\bibitem{fur84} F\"{u}redi, Z. and Pal\'{a}sti, I.: Arrangements of lines with a large number of triangles. Proc. Amer. Math. Soc.  {\bf 92},  561-566 (1984).
 \bigskip\par  
\bibitem{gah15} G\"{a}hler, F., Kwan E.E. and Maloney, G.R.: A computer search for planar substitution tilings with n-fold rotational symmetry. Discrete Comput.Geom. Vol.53, n.2, 445-465 . (2015)
 \bigskip\par
 \bibitem{god89} Godr\`{e}che, C. and Luck, J.M.: Quasiperiodicity and randomness in tilings of the plane.  J.Stat.Phys {\bf 55}, 1-28 (1989).
     \bigskip\par 
\bibitem{gru87} Gr\"{u}nbaum, B. and Shephard, G.C. : Tilings and Patterns. Freeman. New York. (1987).
 \bigskip\par 
\bibitem{mal15} Maloney, G.R.: On substitution tilings of the plane with n-fold rotational symmetry. Discrete Math. Theor. Comput. Sci.  {\bf 17}, 395-412 (2015).
    \bigskip\par 
\bibitem{mey70} Meyer, Y.: Nombres de Pisot, Nombres de Salem et Analyse Harmonique, Lect. Notes Math. {\bf 117} Springer, Berlin (1970).
    \bigskip\par 
\bibitem{mey72} Meyer, Y.: Algebraic Numbers and Harmonic Analysis, North Holland, Amsterdam (1972).
  \bigskip\par 
\bibitem{nis96} Nischke, K.P. and  Danzer, L.: A construction of inflation rules  based on n-fold symmetry.  Discrete Comput.Geom.  {\bf 15}, 221-236  (1996).
 \bigskip\par 
\bibitem{pen74} Penrose, R.: The role of aesthetics in pure and applied mathematical research. Bull. Inst. Math. Appl.   {\bf 10}, 266-271 (1974).
    \bigskip\par 
\bibitem{sad08} Sadun, L.: Topology of tiling spaces. University Lecture Series. Vol. 46.  Amer. Math. Soc. Providence, RI (2008).
 \bigskip\par 
\bibitem{she84} Shechtman, D., Blech, I., Gratias, D., and Cahn, J.W.: Metallic phase with long-range orientational order. Phys. Rev. Lett.  {\bf 53}, 1951-1953 (1984).
 \bigskip\par 
\bibitem{wan61} Wang H.: Proving theorems by pattern recognition-II. Bell System Technical Journal  {\bf 40}, 1-41 (1961).
\end{thebibliography}
\end{document}